\documentclass[11pt]{article}
\usepackage[top=1.25in,bottom=1.25in,left=1.25in,right=1.25in]{geometry}

\usepackage{amssymb, amsmath, amsthm, amscd, amsfonts, graphicx, bm}
\usepackage{newtxmath}
\usepackage{mathtools}
\usepackage{prettyref}
\usepackage{enumerate}
\usepackage{algorithm}
\usepackage{authblk}
\usepackage[noend]{algpseudocode}
\usepackage{pgfplots}
\usepackage{dsfont}
\usepackage{threeparttable}
\usepackage[subpreambles=true]{standalone} 
\usepackage{import} 
\usepackage[toc,page]{appendix}
\usepackage[labelsep=period]{caption, subcaption} 
\usepackage[authoryear, square]{natbib} 
\RequirePackage[colorlinks,citecolor=blue,urlcolor=blue]{hyperref}
\usepackage{cases}

\usepackage{algorithm}
\usepackage{algpseudocode}
\usepackage{diagbox}
\usepackage{multirow}
\usepackage{makecell}

\usepackage{enumitem}

\theoremstyle{plain}
\newtheorem{theorem}{Theorem}[section]

\newtheorem{lemma}{Lemma}[section]
\newtheorem{proposition}{Proposition}

\theoremstyle{definition}
\newtheorem{definition}{Definition}

\newtheorem{assumption}{Assumption}

\theoremstyle{remark}

\let\vec=\mathbf

\newcommand{\vct}[1]{\bm{#1}}
\newcommand{\mtx}[1]{\bm{#1}}
\renewcommand{\vec}[1]{{\boldsymbol{#1}}}

\newcommand{\rank}{\operatorname{rank}}

\newcommand{\Diag}{\operatorname{Diag}}

\renewcommand{\P}{\operatorname{\mathbb{P}}}
\newcommand{\E}{\operatorname{\mathbb{E}}}
\newcommand{\var}{\textrm{var}}
\newcommand{\Var}{\textrm{Var}}

\newcommand{\diag}{\text{diag}}

\newcommand{\N}{\mathcal{N}}

\newcommand{\widesim}[2][1.5]{
  \mathrel{\overset{#2}{\scalebox{#1}[1]{$\sim$}}}
}

\DeclareFontFamily{U}{mathx}{\hyphenchar\font45}
\DeclareFontShape{U}{mathx}{m}{n}{
      <5> <6> <7> <8> <9> <10>
      <10.95> <12> <14.4> <17.28> <20.74> <24.88>
      mathx10
      }{}
\DeclareSymbolFont{mathx}{U}{mathx}{m}{n}
\DeclareFontSubstitution{U}{mathx}{m}{n}
\DeclareMathAccent{\widecheck}{0}{mathx}{"71}
\DeclareMathAccent{\wideparen}{0}{mathx}{"75}

\definecolor{xl}{RGB}{200,50,120}

\begin{document}

\title{Selecting the Number of Communities for Weighted Degree-Corrected Stochastic Block Models}
\author{Yucheng Liu}
\author{Xiaodong Li}
\affil{Department of Statistics, University of California, Davis,  Davis, CA 95616, USA}

\date{}       

\maketitle

\begin{abstract}
We investigate how to select the number of communities for weighted networks without a full likelihood modeling. First, we propose a novel weighted degree-corrected stochastic block model (DCSBM), where the mean adjacency matrix is modeled in the same way as in the standard DCSBM, while the variance profile matrix is assumed to be related to the mean adjacency matrix through a given variance function. Our method of selecting the number of communities is based on a sequential testing framework. In each step, the weighted DCSBM is fitted via some spectral clustering method. A key component of our method is matrix scaling on the estimated variance profile matrix. The resulting scaling factors can be used to normalize the adjacency matrix, from which the test statistic is then obtained. Under mild conditions on the weighted DCSBM, our proposed procedure is shown to be consistent in estimating the true number of communities. Numerical experiments on both simulated and real-world network data demonstrate the desirable empirical properties of our method.
\end{abstract}


\section{INTRODUCTION}
\label{sec:setup}

In network analysis, community detection is regarding how to partition the vertices of a graph into clusters with similar connection patterns. This problem has numerous applications in a variety of fields, including applied physics, sociology, economics, and biology; see the survey paper \citet{fortunato2010community}. An important problem in community detection is selecting the number of communities. To answer this question, various methods have been proposed for unweighted networks, such as sequential testing \citep[]{lei2016goodness, jin2022optimal, han2023universal}, spectral thresholding \citep{le2015estimating}, and model comparison based methods \citep[]{daudin2008mixture, wang2017likelihood, saldana2017many, chen2018network, li2020network, hu2020corrected, ma2021determining}, etc. In particular, sequential testing methods are usually based on goodness-of-fit testing methods \citep[e.g.,][]{gao2017testing, jin2018network, jin2021optimal, hu2021using}.

Most existing methods for selecting the number of communities are only applicable to unweighted networks. In contrast, the problem of determining the number of communities for weighted networks is little studied. In fact, weighted networks are common in practice and often reveal more refined community structures; see, for example, \citet{newman2004finding}, \citet{ball2011efficient} and \citet{newman2016estimating}. A challenge for rank selection in weighted networks lies in the modeling of the networks, since likelihood based methods could be very restrictive for this problem.

In this paper, we propose a generic model for weighted networks without modeling the likelihood. To be specific, we only model the mean structure in the form of the degree-corrected stochastic block model (DCSBM) as in \citet{karrer2011stochastic}, and the entrywise variances as functions of the corresponding means. 
Due to its connection to the standard DCSBM, our proposed model is referred to as the weighted DCSBM.
To see why the variance profile matrix is helpful for selecting the number of communities, note that the adjacency matrix can be decomposed into $\mtx{A} = \mtx{M} + \left(\mtx{A} -\mtx{M}\right)$, where $\mtx{M}$ is the mean adjacency matrix. If we denote $K$ as the number of communities, then $\mtx{M}$ is of rank $K$, and $\mtx{A}$ can be viewed as a perturbation of a rank-$K$ matrix. Denote the nonzero eigenvalues of $\mtx{M}$ as $|\lambda_1(\mtx{M})| \geq \cdots \geq |\lambda_K(\mtx{M})|>0$, and all eigenvalues of $\mtx{A}$ as $|\lambda_1(\mtx{A})| \geq \cdots \geq |\lambda_n(\mtx{A})|$. Assume that we know an explicit and tight upper bound $\|\mtx{A} - \mtx{M}\| \leq \tau$. Then, we have $|\lambda_{K+1}(\mtx{A})| \leq \|\mtx{A} - \mtx{M}\| \leq \tau$. Furthermore, if we assume the signal-to-noise ratio of the weighted DCSBM is large enough, which is $|\lambda_K(\mtx{M})| \gg \tau$, then $|\lambda_{K}(\mtx{A})| \geq |\lambda_{K}(\mtx{M})| - \|\mtx{A} - \mtx{M}\| \gg \tau$. Then in step $m =1, 2, ...$ of a sequential testing procedure, we can choose the test statistic as $T_{n, m} = |\lambda_{m+1}(\mtx{A})|$, which satisfies $T_{n, m} \gg \tau$ for $m \leq K-1$, and $T_{n, m} \leq \tau$ for $m=K$. Consequently, we can use $T_{n, m} \leq \tau$ as the stopping rule to select the number of communities.

However, such an explicit and tight bound $\|\mtx{A} - \mtx{M}\| \leq \tau$ is generally not available. One natural idea is to consider a normalized version of the noise matrix, which allows us to derive a tight and explicit upper bound for its spectral norm. Here we borrow an idea from \citet{landa2022biwhitening}, which studies how to reveal the rank of a Poisson data matrix by spectral truncation with the assistance of matrix scaling on the variance profile matrix. Returning to our problem, if the variance profile matrix $\mtx{V}$ is known, its entries are $V_{ij}= \var(A_{ij}) = \var(A_{ij} - M_{ij})$. It is known that there exists a diagonal matrix $\mtx{\Psi}  = \Diag(\psi_1, \ldots, \psi_n)$, such that $\mtx{\Psi} \mtx{V}\mtx{\Psi}$ is doubly stochastic, i.e., every row sum of $\mtx{\Psi} \mtx{V}\mtx{\Psi}$ is $1$ \citep[e.g.,][]{knight2014symmetry}. With this diagonal scaling matrix $\mtx{\Psi}$, recent results \citep[e.g.,][]{latala2018dimension} in random matrix theory guarantee that $\left\|\mtx{\Psi}^{\frac{1}{2}}(\mtx{A} - \mtx{M})\mtx{\Psi}^{\frac{1}{2}} \right\| \leq 2+\epsilon$ for some small positive constant $\epsilon$.

Based on the above heuristic, we propose the following sequential testing procedure. For each $m=1, 2, 3, \ldots$, we first group the nodes into $m$ distinct communities using some spectral clustering method, e.g., SCORE proposed in \citet{jin2015fast} or the regularized spectral clustering (RSC) procedure proposed in \citet{amini2013pseudo}. With the estimated groups, we obtain the estimated mean adjacency matrix $\widehat{\mtx{M}}^{(m)}$ by fitting the DCSBM, and further derive the estimated variance profile matrix $\widehat{\mtx{V}}^{(m)}$ using the variance-mean relationship. Next, we find a diagonal scaling matrix $\widehat{\mtx{\Psi}}^{(m)}$ such that $\widehat{\mtx{\Psi}}^{(m)} \widehat{\mtx{V}}^{(m)} \widehat{\mtx{\Psi}}^{(m)}$ is doubly stochastic. The test statistic is defined as $T_{n, m} = \left|\lambda_{m+1} \left(\left(\widehat{\mtx{\Psi}}^{(m)} \right)^{\frac{1}{2}} \mtx{A} \left(\widehat{\mtx{\Psi}}^{(m)} \right)^{\frac{1}{2}}\right)\right|$. For each $m$, we stop the iterative procedure if $T_{n, m} < 2+\epsilon$, and then report the estimated number of communities as $\widehat{K}=m$.

\subsection{Related Work}
The stochastic block model (SBM) by \citet{holland1983stochastic} and its variants, such as the degree-corrected stochastic block model (DCSBM) by \citet{karrer2011stochastic}, 
have been widely used to model community structures in networks. Many community detection methods have been proposed in the literature, with examples including modularity and likelihood based methods \citep[e.g.,][]{newman2004finding, bickel2009nonparametric, zhao2012consistency, amini2013pseudo, bickel2013asymptotic, le2016optimization, zhang2016minimax, zhang2020theoretical}, spectral methods \citep[e.g.,][]{rohe2011spectral, amini2013pseudo, jin2015fast, lei2015consistency, joseph2016impact, gulikers2017spectral, gao2018community, abbe2020entrywise, zhang2020detecting, deng2021strong, li2022hierarchical, lei2020consistency}, convex optimization methods \citep[e.g.,][]{cai2015robust, Vershynin14, abbe2015exact, chen2018convexified, li2021convex}, and many other methods such as Bayesian approaches.

Matrix scaling has been used in \citet{landa2022biwhitening} to reveal the rank of a Poisson data matrix, where the variance profile matrix is the same as the mean adjacency matrix and thus can be estimated by the observed count data matrix. In contrast, our work is not restricted to the Poisson data matrix. In fact, we have a parametric model for the mean structure in our weighted DCSBM, whereas there is no such parametric mean modeling in \citet{landa2022biwhitening}.

Sequential testing based on stepwise SBM/DCSBM fitting has been commonly used in the literature of rank selection for standard binary networks \citep[e.g.,][]{lei2016goodness, wang2017likelihood, hu2020corrected, ma2021determining, hu2021using, jin2022optimal}. However, these methods either rely on the binary structure of the unweighted network or on modeling the complete likelihood of the random network, and thus cannot be straightforwardly extended to our setup.

To address some technical challenges in the underfitting case $m<K$, we generalize the Nonsplitting Property established in \citet{jin2022optimal} for the spectral clustering approach SCORE proposed in \citet{jin2015fast} to the weighted DCSBM. This is the primary reason why our main results are established using SCORE for spectral clustering, although extensions to other spectral methods may be possible, which we leave for future investigation.

\subsection{Organization of the Paper}
This paper is organized as follows. In Section \ref{sec:DCSBM}, we extend the standard DCSBM to model the mean structure of weighted networks, and link the mean adjacency matrix and the variance profile matrix with a variance function. Section \ref{sec:methodology} introduces a stepwise DCSBM fitting procedure with a spectral test statistic based on variance profile matrix scaling. The main results on the consistency of our procedure are presented in Section \ref{sec:theory} under certain assumptions on the weighted DCSBM. In Section \ref{sec:experiments}, we demonstrate the empirical properties of our proposed procedure with both synthetic and real-world weighted networks. The proofs of our main results are given in Appendix \ref{sec:proof_main}. Preliminary results and proofs of the technical lemmas are also included in the supplementary material.

\section{WEIGHTED DEGREE-CORRECTED STOCHASTIC BLOCK MODELS}
\label{sec:DCSBM}
Our first task is to extend the standard DCSBM proposed in \citet{karrer2011stochastic} to a weighted DCSBM which accommodates networks with nonnegative weights. The edges in our weighted DCSBM are assumed to be independent. Rather than specifying the exact likelihood of the random graph with weighted edges as in the standard DCSBM, the weighted DCSBM relies on specifying the first two moments, namely, the mean adjacency matrix and the variance profile matrix. 

We begin by introducing some notation for the network. Let $\mtx{A}$ be the adjacency matrix of the weighted network with $n$ nodes, which belong to $K$ separate communities. Denote $\mathcal{N}_1, \ldots, \mathcal{N}_K$ as the underlying communities, with respective cardinalities $n_1, \ldots, n_K$, and thereby $n = n_1 + \cdots + n_K$. Let $\phi:[n]\rightarrow[K]$ be the community membership function of node, such that $\phi(i) = k$ if and only if node $i$ belongs to community $\mathcal{N}_k$. We also use the vector $\vct{\pi}_i \in \mathbb{R}^K$ to indicate the community belonging of node $i$ by letting $\pi_{i}(k)=1$ if $i \in \mathcal{N}_k$ and $\pi_{i}(k)=0$ otherwise. Denote $\mtx{\Pi} = [\vct{\pi}_1, \ldots, \vct{\pi}_n]^\top \in \mathbb{R}^{n \times K}$ as the community membership matrix.

The mean adjacency matrix is modeled similarly as in the standard DCSBM. Note that the network is allowed to have self-loops for simplicity of analysis without loss of generality. Specifically, let $\mtx{B}$ be a $K \times K$ symmetric community connectivity matrix with positive entries and diagonal entries $B_{kk}=1$ for $k=1,\ldots, K$
to ensure identifiability. Also, let $\theta_1, \ldots, \theta_n >0$ be the heterogeneity parameters. Denote $\vct{\theta} = (\theta_1, \ldots, \theta_n)^\top$ and  $\mtx{\Theta}=\diag (\theta_1, \ldots, \theta_n)$. 
Then, the entries of the expected mean matrix are parameterized as
\begin{equation}
\label{eq:DCSBM}
\E[A_{ij}] \coloneqq M_{ij} \coloneqq \theta_i \theta_j B_{\phi(i) \phi(j)} = \theta_i \theta_j \vct{\pi}_i^\top \mtx{B} \vct{\pi}_j, 
\end{equation}
for $1\leq i \leq j \leq n.$
In matrix form, the mean adjacency matrix can be represented as
\begin{equation}
\label{eq:DCSBM_mean}
\mtx{M} = \mtx{\Theta} \mtx{\Pi} \mtx{B} \mtx{\Pi}^\top \mtx{\Theta}.
\end{equation}

In modeling the second order moments, we denote $V_{ij}= \var(A_{ij})$ and assume that it is a function of the corresponding mean by some variance function $\nu(\cdot)$, i.e., $V_{ij} = \nu(M_{ij})$. By denoting the variance profile matrix $\mtx{V} = [V_{ij}]_{i, j = 1}^n$, we can also represent the variance-mean relationship as $\mtx{V} = \nu (\mtx{M})$, where $\nu$ is viewed as an entrywise operation.

The variance function $\nu(\cdot)$ is known for some common distributions, such as $\nu(\mu) = \mu (1 - \mu)$ for the Bernoulli case and $\nu(\mu) = \mu$ for the Poisson case. In the present work, we assume $\nu$ is known. A potential direction for future study is how to estimate $\nu$ in a data-dependent manner, but we leave this for future investigation.

\section{METHODOLOGY}
\label{sec:methodology}
In this section, we will outline the details of our stepwise procedure for estimating the number of communities in weighted networks, referred to as Stepwise Variance Profile Scaling (SVPS). SVPS is a stepwise method that relies on fitting the weighted DSCBM with an increasing sequence of candidate numbers of communities. In other words, for each positive integer $m = 1, 2, \ldots$, we aim to test $H_0: K=m$, where $K$ is the true rank of the weighted DCSBM.

\subsection{Stepwise Weighted DCSBM Fitting}
\label{sec:dcsbm_fitting}

With each hypothetical number of communities $m$, we apply a standard community detection method, such as SCORE \citep{jin2015fast} or RSC \citep{amini2013pseudo, joseph2016impact}, to obtain $m$ estimated communities $\widehat{\mathcal{N}}_1^{(m)}, \ldots, \widehat{\mathcal{N}}_m^{(m)}$. Subsequently, the DCSBM parameters $\vct{\theta}$, $\mtx{B}$ and $\mtx{W}$ can be directly estimated in a plug-in manner.

The estimates of $\vct{\theta}$ and $\mtx{B}$ are the same as in the standard DCSBM, and here we give a review of their derivations in \citet{jin2022optimal}. Let us first see how to represent $\vct{\theta}$ and $\mtx{B}$ with the mean adjacency matrix $\mtx{M}$, the true communities $\mathcal{N}_1, \ldots, \mathcal{N}_K$, and the expected degrees. Decompose $\vct{\theta}$ as $\vct{\theta} = \vct{\theta}_1 + \cdots + \vct{\theta}_K$, where $\vct{\theta}_k \in \mathbb{R}^n$ for $k=1, \ldots, K$ such that $\vct{\theta}_k(i) = \theta_i$ if $i \in \mathcal{N}_k$ and $\vct{\theta}_k(i) = 0$ otherwise. We can similarly decompose the $n$-dimensional all-one vector into $\vct{1}_n = \vct{1}_1 + \cdots + \vct{1}_K$ such that $\vct{1}_k(j) = 1$ if $j \in \mathcal{N}_k$ and $\vct{1}_k(j) = 0$ otherwise. 
It is easy to verify that for $1 \leq k, l \leq K$, $\vct{1}_k^\top \mtx{M} \vct{1}_l = B_{kl} \|\vct{\theta}_k\|_1 \|\vct{\theta}_l\|_1$. By the assumption $B_{kk} = 1$ for $k=1, \ldots, K$, the above equality implies $\|\vct{\theta}_k\|_1 = \sqrt{\vct{1}_k^\top \mtx{M} \vct{1}_k}$,
which further gives
\begin{equation}
\label{eq:B_kl}
B_{kl} = \frac{\vct{1}_k^\top \mtx{M} \vct{1}_l}{\|\vct{\theta}_k\|_1 \|\vct{\theta}_l\|_1} = \frac{\vct{1}_k^\top \mtx{M} \vct{1}_l}{\sqrt{\vct{1}_k^\top \mtx{M} \vct{1}_k}\sqrt{\vct{1}_l^\top \mtx{M} \vct{1}_l}}.
\end{equation}

Denote the degree of node $i \in \mathcal{N}_k$ as $d_i = \sum_{j=1}^n A_{ij}$. Its expectation, referred to as the population degree, is given by
\begin{align*}
d_i^* & \coloneqq \E[d_i] = \theta_i \left(B_{k1}\|\vct{\theta}_1\|_1 + B_{k2}\|\vct{\theta}_2\|_1 + \cdots + B_{kK}\|\vct{\theta}_K\|_1\right)
\\
&= \theta_i \left(\frac{\vct{1}_{k}^\top \mtx{M} \vct{1}_1}{\sqrt{\vct{1}_{k}^\top \mtx{M} \vct{1}_k}} + \cdots + \frac{\vct{1}_k^\top \mtx{M} \vct{1}_K}{\sqrt{\vct{1}_k^\top \mtx{M} \vct{1}_k}} \right)
= \theta_i \frac{\vct{1}_k^\top \mtx{M} \vct{1}_n}{\sqrt{\vct{1}_k^\top \mtx{M} \vct{1}_k}}.
\end{align*}
This implies that the degree-correction parameter $\theta_i$ can be expressed as
\begin{equation}
\label{eq:theta}
\theta_i = \frac{\sqrt{\vct{1}_k^\top \mtx{M} \vct{1}_k}}{\vct{1}_k^\top \mtx{M} \vct{1}_n} d_i^*, \quad i \in \mathcal{N}_k.
\end{equation}

With \eqref{eq:B_kl} and \eqref{eq:theta}, we obtain plug-in estimates of $\vct{\theta}$ and $\mtx{B}$ by replacing the true community partition $\mathcal{N}_1, \ldots, \mathcal{N}_K$ with the estimated partition $\widehat{\mathcal{N}}_1^{(m)}, \ldots, \widehat{\mathcal{N}}_m^{(m)}$, replacing $\mtx{M}$ with $\mtx{A}$, and replacing $d_i^*$ with $d_i$. In analogy to $\vct{1}_n = \vct{1}_1 + \cdots + \vct{1}_K$, we decompose the all-one vector as the sum of indicator vectors corresponding to the estimated communities:
$\vct{1}_n = \hat{\vct{1}}_1^{(m)} + \cdots + \hat{\vct{1}}_m^{(m)}$, where for each $j=1, \ldots, n$ and $k=1, \ldots, m$, $\hat{\vct{1}}_k^{(m)}(j)=1$ if $j \in \widehat{\mathcal{N}}_k^{(m)}$ and $\hat{\vct{1}}_k^{(m)}(j)=0$ otherwise. Then, the plug-in estimates are
\begin{equation}
\label{eq:theta_hat}
\hat{\theta}^{(m)}_i \coloneqq \frac{\sqrt{(\hat{\vct{1}}_k^{(m)})^\top \mtx{A}\hat{\vct{1}}_k^{(m)}}}{(\hat{\vct{1}}_k^{(m)})^\top \mtx{A} \vct{1}_n} d_i, \text{~}  k=1,\ldots, m \text{~and~} i \in \widehat{\mathcal{N}}_k^{(m)},
\end{equation}
and
\begin{align}
\label{eq:B_hat}
    & \quad \quad \widehat{B}_{kl}^{(m)} \coloneqq \frac{(\hat{\vct{1}}^{(m)}_k)^\top \mtx{A} \hat{\vct{1}}^{(m)}_l}{{\|\hat{\vct{\theta}}^{(m)}_k\|_1 \|\hat{\vct{\theta}}^{(m)}_l\|_1}} \nonumber \\
    = & \frac{(\hat{\vct{1}}^{(m)}_k)^\top \mtx{A} \hat{\vct{1}}^{(m)}_l}{\sqrt{(\hat{\vct{1}}^{(m)}_k)^\top \mtx{A} \hat{\vct{1}}^{(m)}_k}\sqrt{(\hat{\vct{1}}^{(m)}_l)^\top \mtx{A} \hat{\vct{1}}^{(m)}_l}}, \text{~}  1\leq k, l \leq m.
\end{align}
Then we have the estimated mean matrix
\begin{equation}
\label{eq:M_hat}
\widehat{M}_{ij}^{(m)} \coloneqq \hat{\theta}_i^{(m)}\hat{\theta}_j^{(m)} \widehat{B}_{kl}^{(m)} = \frac{(\hat{\vct{1}}^{(m)}_k)^\top \mtx{A} \hat{\vct{1}}^{(m)}_l}{(\hat{\vct{1}}_k^{(m)})^\top \mtx{A} \vct{1}_n  \cdot  (\hat{\vct{1}}_l^{(m)})^\top \mtx{A} \vct{1}_n } d_i d_j,
\end{equation}
for any $1 \leq i, j \leq n$ with $i \in  \widehat{\mathcal{N}}_k^{(m)}$ and $j \in  \widehat{\mathcal{N}}_l^{(m)}$. 
By denoting the community membership matrix corresponding to $\widehat{\mathcal{N}}_1^{(m)}, \ldots, \widehat{\mathcal{N}}_m^{(m)}$ as $\widehat{\mtx{\Pi}}^{(m)} \in \mathbb{R}^{n \times m}$, the matrix form of mean estimation is (omitting the superscripts)
\begin{equation}
\label{eq:M_hat_matrix}
\widehat{\mtx{M}}^{(m)} = \widehat{\mtx{\Theta}} \widehat{\mtx{\Pi}} \widehat{\mtx{B}} \widehat{\mtx{\Pi}}^\top \widehat{\mtx{\Theta}}.
\end{equation}

With the estimated mean matrix $\widehat{\mtx{M}}^{(m)}$, we have the residual matrix $\mtx{R}^{(m)} = \mtx{A} - \widehat{\mtx{M}}^{(m)}$ and the estimated variance profile matrix $\widehat{\mtx{V}}^{(m)} = \nu (\widehat{\mtx{M}}^{(m)})$. In the case where the variance function is unknown, we substitute $\nu(\cdot)$ with the estimated variance function $\widehat{\nu}(\cdot)$.

\subsection{Variance Profile Scaling and Spectral Statistic}
\label{sec:matrix_scaling}
In Section \ref{sec:setup}, we have introduced how to define the spectral test statistic based on variance profile matrix scaling. Here, we briefly introduce it again. For the variance profile matrix $\mtx{V}$, assume $\mtx{\Psi}  = \Diag(\vct{\psi}) = \Diag(\psi_1, \ldots, \psi_n)$ is a diagonal matrix such that $\mtx{\Psi} \mtx{V}\mtx{\Psi}$ is doubly stochastic, i.e.,
\begin{equation}
\label{eq:pop_scaling}
\sum_{i=1}^n V_{ij} \psi_i \psi_j = 1, \quad \forall j=1,\ldots, n.
\end{equation}
The uniqueness and existence of such $\mtx{\Psi}$ and related algorithms are well-known in the literature; see, e.g., \citet{sinkhorn1967diagonal, knight2014symmetry}.

Analogously, in the $m$-th step of SVPS, we choose the scaling factors $\hat{\psi}_1^{(m)}, \ldots, \hat{\psi}_n^{(m)}$, such that $\widehat{\mtx{\Psi}}^{(m)} \widehat{\mtx{V}}^{(m)} \widehat{\mtx{\Psi}}^{(m)}$ is doubly stochastic,
where $\widehat{\mtx{\Psi}}^{(m)} = \Diag(\hat{\psi}_1^{(m)}, \ldots, \hat{\psi}_n^{(m)})$. In other words, 
\begin{equation}
\label{eq:sample_scaling}
\sum_{i=1}^n \widehat{V}_{ij} \hat{\psi}_i \hat{\psi}_j = 1, \quad \forall j=1,\ldots, n.
\end{equation}
Define the test statistic as 
\[
T_{n, m} = \left|\lambda_{m+1} \left(\left(\widehat{\mtx{\Psi}}^{(m)} \right)^{\frac{1}{2}} \mtx{A} \left(\widehat{\mtx{\Psi}}^{(m)} \right)^{\frac{1}{2}}\right)\right|.
\]
For $m=1, 2, \ldots$, we stop the iterative procedure if $T_{n, m} < 2+\epsilon$ for some prespecified small constant $\epsilon$, and then obtain the estimated number of communities as $\widehat{K}=m$.

\section{MAIN RESULTS}
\label{sec:theory}
In this section, we aim to establish the consistency of SVPS in selecting the number of communities under the weighted DCSBM. Our main results consist of two parts: (1) in the null case $m=K$, we will show that $T_{n, m} \leq 2 + o_p(1)$; (2) in the underfitting case $m<K$, we will show that $T_{n, m} \gg O_p(1)$. Before presenting the main results, we first introduce a sequence of conditions on the weighted DCSBM under which our consistency analysis is conducted.

\subsection{Assumptions}
Note that the conditions imposed on the weighted DCSBM are captured by a constant $c_0$.

\begin{assumption}
\label{ass:DCSBM}
Consider the weighted DCSBM described in Section \ref{sec:DCSBM}. Denote $\theta_{\max} = \max\{\theta_1, \ldots, \theta_n\}$ and $\theta_{\min} = \min\{\theta_1, \ldots, \theta_n\}$. Assume the following conditions hold:
\begin{itemize}
\item~[Fixed rank] The true number of communities $K$ is fixed.
\item~[Balancedness] 
\begin{equation}
\label{eq:balance_ass}
    \min_{1 \leq k \leq K} \frac{n_k}{n} \geq c_0 \quad  \text{and} \quad \frac{\theta_{\min}}{\theta_{\max}} \geq c_0. 
\end{equation}

\item~[Sparseness] 
\begin{equation}
\label{eq:sparseness}
\frac{1}{c_0} \geq \theta_{\max} \geq \theta_{\min} \geq \frac{\log^3 n }{\sqrt{n}}. 
\end{equation}

\item~[Community connectivity]
The $K \times K$ matrix $\mtx{B}$ is fixed, and its entries and eigenvalues satisfy
\begin{equation}
\label{eq:B_entry_ass}
\begin{cases}
B_{kk} =1 \text{~~~for~~} k=1, \ldots, K, 
\\
c_0 \leq B_{kl} \leq 1 \text{~~for~} 1\leq k, l \leq K,
\\
\lambda_1(\mtx{B}) > |\lambda_2(\mtx{B})| \geq \cdots \geq |\lambda_K(\mtx{B})| \geq c_0 > 0.
\end{cases}
\end{equation}

\item~[Variance-mean function] The function $\nu(\cdot)$ satisfies
\begin{equation}
\label{eq:variance_linearity}
c_0 \mu \leq \nu(\mu) \leq \mu/c_0
\quad \text{and} \quad
\text{$\nu(\cdot)$ is $1/c_0$-Lipschitz}.
\end{equation}

\item~[Bernstein condition]
For any $i \leq j$ and any integer $p \geq 2$, there holds
\begin{equation}
\label{eq:bernstein}
\E \left[|A_{ij} - M_{ij}|^p \right] \leq \left(\frac{p!}{2}\right) R(c_0)^{p-2} \nu(M_{ij}),
\end{equation}
where $R(c_0)$ is a constant only depending on $c_0$.

\end{itemize}
\end{assumption}

In fact, \eqref{eq:bernstein} is the standard Bernstein condition, which holds for various common distributions. As an example, the following result shows that the Poisson distribution satisfies this condition.
\begin{proposition}
\label{lem:poisson_moment_bd}
Let $X \sim \text{Poisson}(\lambda)$, where $\lambda \leq C(c_0)$ and $C(c_0)$ is a constant only depending on $c_0>0$. Then, for any integer $p \geq 2$, there holds
\[
\E \left[|X - \lambda|^p \right] \leq \left(\frac{p!}{2}\right) R(c_0)^{p-2} \lambda,
\]
where $R(c_0)$ is a constant only depending on $c_0$.
\end{proposition}
\begin{proof}
Since a Poisson random variable is known to be discrete log-concave, by noting that $\Var(X) = \lambda$, this lemma can be directly obtained from Lemma 7.5 and Definition 1.2 of \citet{schudy2011bernstein}.
\end{proof}

\subsection{Consistency}
In the theoretical results presented in this section, we use SCORE \citep{jin2015fast} for spectral clustering. Here, we briefly introduce its implementation. First, compute the $m$ leading singular vectors $\vct{u}_1, \ldots, \vct{u}_m$ of $\mtx{A}$ corresponding to the $m$ largest eigenvalues in magnitude; next, construct a $n \times (m-1)$ matrix of entrywise ratios $\mtx{R}^{(m)}$: $R^{(m)}(i, k) = u_{k+1}(i)/u_1(i)$ for $1\leq i \leq n$ and $0\leq k \leq m-1$; finally, the rows of the ratio matrix $\mtx{R}^{(m)}$ are clustered by the $k$-means algorithm, assuming there are $m$ clusters. In future work, it would be interesting to establish consistency results for other commonly used spectral clustering methods, e.g., the regularized spectral clustering (RSC) methods proposed and studied in \citet{amini2013pseudo, joseph2016impact}.

Next, we introduce the main results that guarantee the consistency of SVPS in selecting the number of communities $K$, provided that Assumption \ref{ass:DCSBM} holds. The following two theorems correspond to two parts of the main results: the null case ($m=K$) and the underfitting case ($m<K$).

\begin{theorem}[Null Case]
\label{thm:null}
If we implement the procedure in Section \ref{sec:methodology} with the candidate number of communities $m=K$ and SCORE as the spectral clustering method, then, for any fixed $c_0>0$ in Assumption \ref{ass:DCSBM}, as $n \rightarrow \infty$, we have $T_{n, m} \leq 2 + o_P(1)$.
\end{theorem}

\begin{theorem}[Underfitting Case]
\label{thm:under-fitting}
If we implement the procedure in Section \ref{sec:methodology} with the candidate number of communities $m<K$ and SCORE as the spectral clustering method, then, for any fixed $c_0 > 0$ in Assumption \ref{ass:DCSBM}, as $n \rightarrow \infty$, we have $T_{n, m} \stackrel{P}{\longrightarrow} \infty$.
\end{theorem}

Obviously, combining Theorems \ref{thm:null} and \ref{thm:under-fitting} shows the consistency of SVPS. In \eqref{eq:balance_ass}, the sparsity assumption $\theta_{\min} \geq \left. (\log^3 n)\middle/\sqrt{n} \right.$ might be suboptimal. In contrast, the corresponding assumption is typically $\theta_{\min} \geq C_0 \sqrt{ (\log n)/n }$ in the literature of unweighted network model selection \citep[e.g.,][]{jin2022optimal}. The extra logarithms for the weighted DCSBM arise from an application of the truncation technique in the proof, since existing spectral radius results for sparse and heterogeneous random matrices in the literature, e.g., \citet{latala2018dimension}, usually require the entries to be uniformly bounded. It is an interesting question whether the logarithm factors can be further improved, but we leave this for future investigation.

\section{NUMERICAL EXPERIMENTS}
\label{sec:experiments}

This section presents numerical experiments using both synthetic and real-world data to demonstrate the empirical behavior of SVPS. In the step of weighted DCSBM fitting in SVPS, we use either SCORE \citep{jin2015fast} or RSC \citep{amini2013pseudo, joseph2016impact}. If RSC is used, the regularization parameter is set to $0.25 \left(\bar{d}/n\right)$, where $\bar{d}$ is the average node degree of the network. In the step of sequential testing, the threshold for comparing $T_{n, m}$ with is $2.02$, $2.05$ or $2.10$. We will specify the exact choice in each experiment.

Throughout all our experiments, we select the variance function in SVPS as $\hat{\nu}(\mu) = \mu$. Note that this does not imply that the true variance function satisfies this relationship. In other words, model mismatching is allowed in our experiments. Therefore, the estimated variance profile matrix always satisfies $\widehat{\mtx{V}}^{(m)} = \widehat{\mtx{M}}^{(m)} =  \widehat{\mtx{\Theta}} \widehat{\mtx{\Pi}} \widehat{\mtx{B}} \widehat{\mtx{\Pi}}^\top \widehat{\mtx{\Theta}}$.

\subsection{Synthetic Networks}
\label{sec:experiment_synthetic}

\subsubsection{Weighted DCSBM Generation}

Now we discuss how to generate weighted networks based on certain weighted DCSBM for our simulations. We consider three types of distributions: Poisson, binomial, and negative binomial. In each case, the distribution of the weighted DCSBM is determined by its mean structure. 
In fact, to determine the mean structure, we only need to specify the $K \times K$ matrix $\mtx{B}$, the vector of degree-correction parameters $\vct{\theta}$, and the block sizes. Therefore, we follow the experiment setting in \citet{hu2020corrected} and specify the above parameters as follows:
\begin{itemize}[topsep=0pt]
\item For $\mtx{B}$, let
\[
B_{kl} = \rho \left(1 + r \times \mathbf{1}_{\{k = l\}} \right),
\]
where $\rho$ is a sparsity parameter and $r$ is a contrast parameter. We consider the following combinations of $(\rho, r)$: $(0.04, 4)$, $(0.06, 3)$ and $(0.12, 2)$ in our simulations. 

\item The block sizes are set according to the sequence $\vct{n}_{all} = (50, 100, 150, 50, 100, 150)$. For example, $(n_1, n_2) = (50, 100)$ when $K=2$, $(n_1, n_2, n_3) = (50, 100, 150)$ when $K=3$, etc. 

\item The degree-correction parameters are $\textit{i.i.d.}$ generated from the following distribution:
\begin{align*}
    \begin{cases}
    \text{Uniform}(0.6, 1.4),\quad &\text{with probability } 0.8; \\
    0.5, \quad &\text{with probability } 0.1; \\
    1.5, \quad &\text{with probability } 0.1.
    \end{cases}
\end{align*}
\end{itemize}

After generating $\mtx{M}$, we consider the following simulations to generate $\mtx{A}$:
\begin{itemize}[noitemsep, topsep=0pt]
    \item[$\blacksquare$] \textbf{Simulation 1: } $A_{ij} \widesim{ind} \text{Poisson}(M_{ij})$. 
    \item[$\blacksquare$] \textbf{Simulation 2: } $A_{ij} \widesim{ind} \text{Binomial}(5, M_{ij}/5)$.
    \item[$\blacksquare$]
    \textbf{Simulation 3: } $A_{ij} \widesim{ind} \text{NB}(5, M_{ij}/5)$, where NB$(n, p)$ is the negative binomial distribution with $n$ successful trials and the probability of success $1-p$ in each trial. 
\end{itemize}


\begin{figure*}[h!]
\begin{minipage}{1\textwidth}
\begin{subfigure}{.32\textwidth}
  \centering
  \includegraphics[width = 1.0\linewidth, height=0.7 \linewidth]{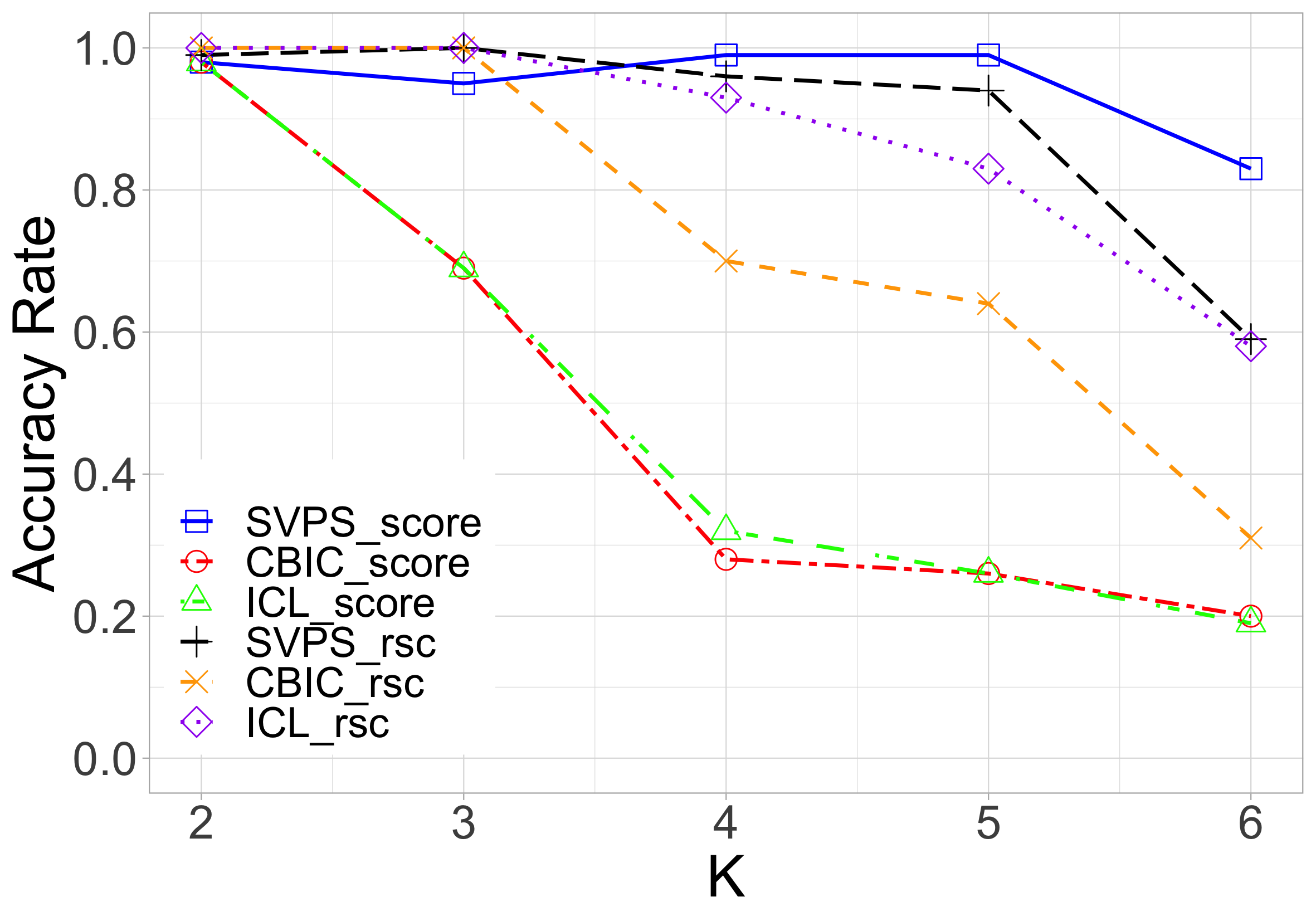}
  \caption{$\rho = 0.04$, $r=4$; $\epsilon = 0.1$}
\end{subfigure}%
\begin{subfigure}{.32\textwidth}
  \centering
  \includegraphics[width = 1.0\linewidth, height=0.7 \linewidth]{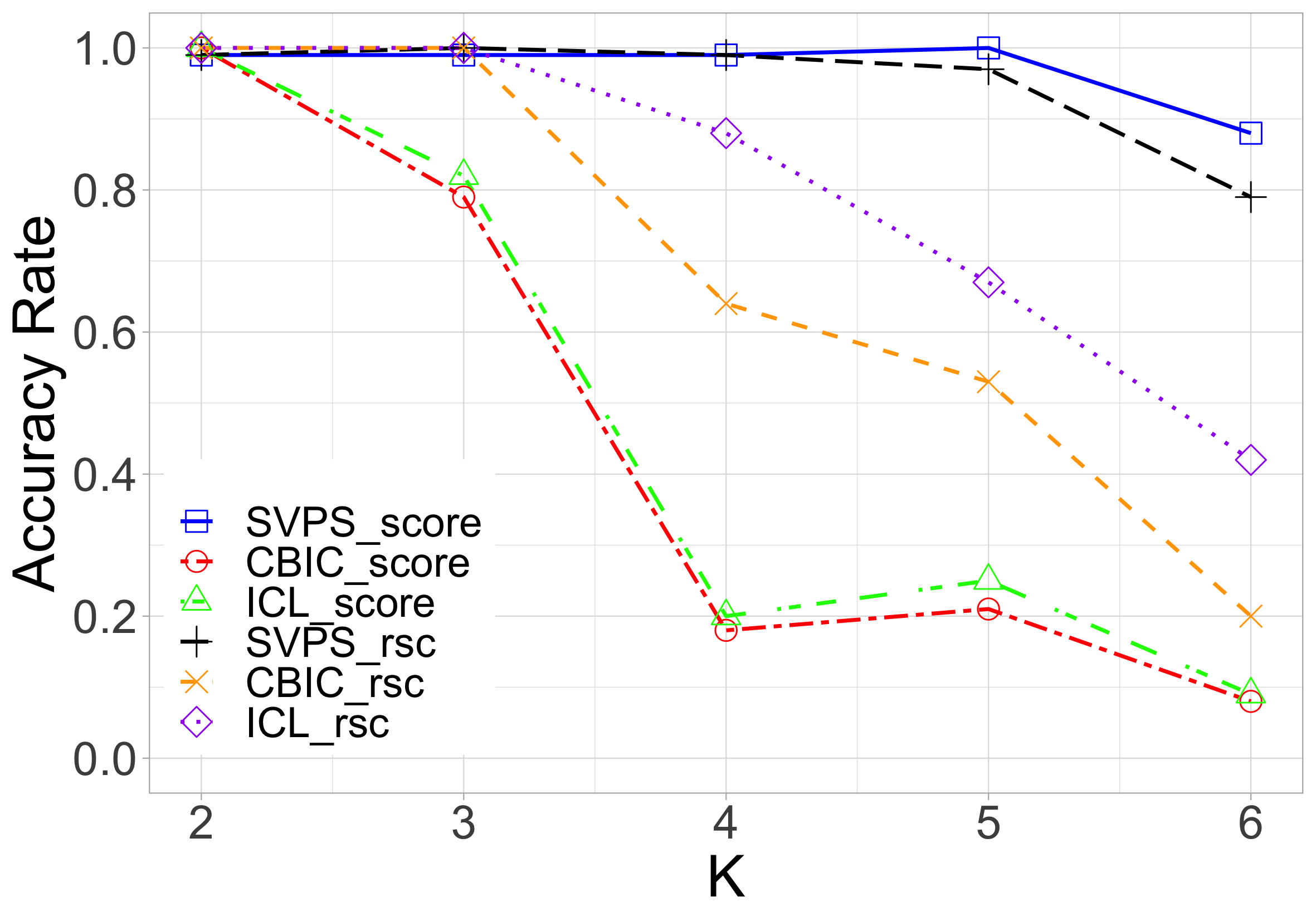}
  \caption{$\rho = 0.06$, $r=3$; $\epsilon = 0.05$}
\end{subfigure}
\begin{subfigure}{.32\textwidth}
  \centering
  \includegraphics[width = 1.0\linewidth, height=0.7\linewidth]{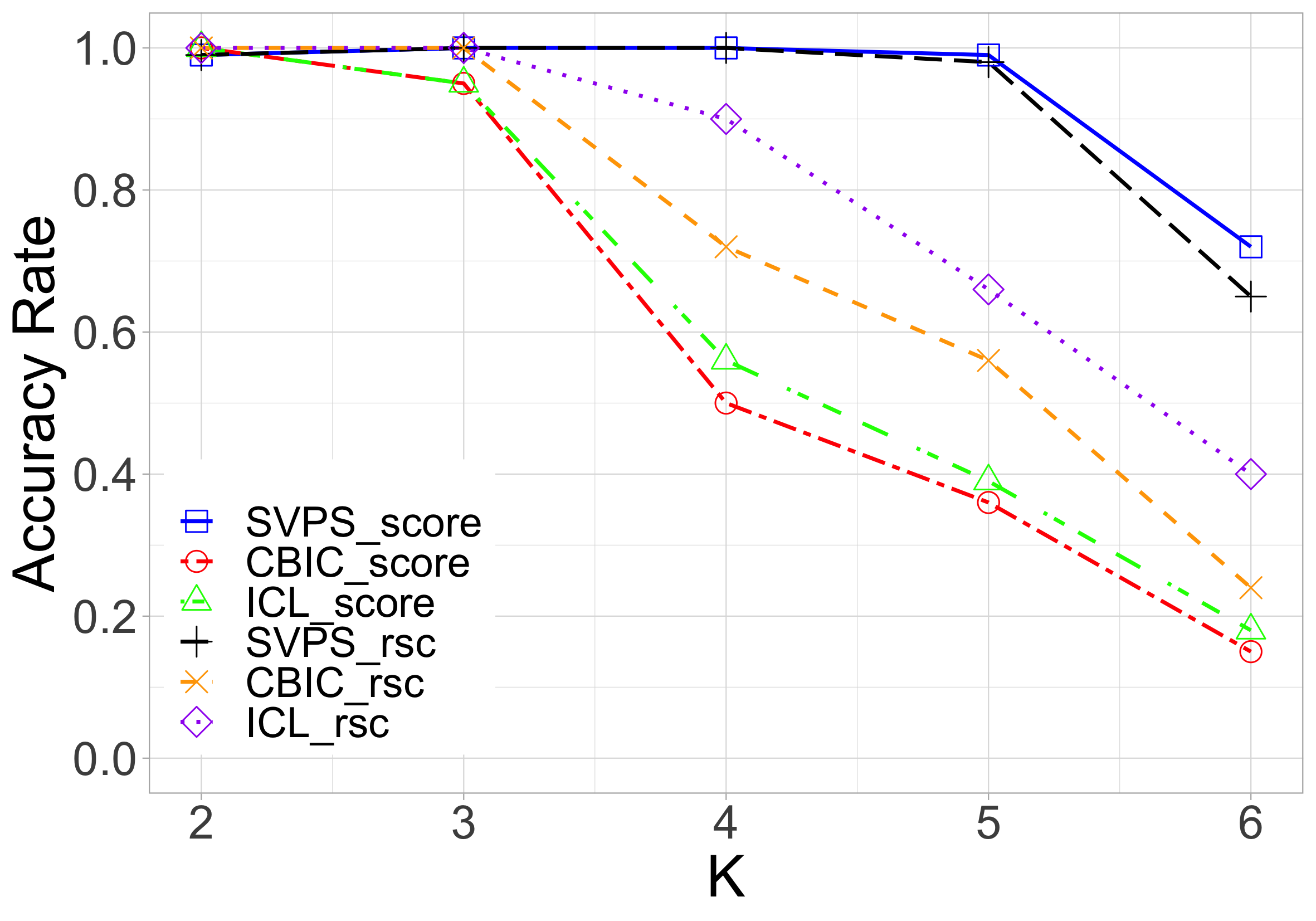}
  \caption{$\rho = 0.12$, $r=2$; $\epsilon = 0.05$}
\end{subfigure}
\hspace{0.003\textwidth}
\raisebox{4.5\height}{\rotatebox[origin=c]{90}{\scalebox{0.8}{Poisson}}}
\end{minipage}%

\begin{minipage}{1\textwidth}
\begin{subfigure}{.32\textwidth}
  \centering
  \includegraphics[width=1.0\linewidth, height=0.7\linewidth]{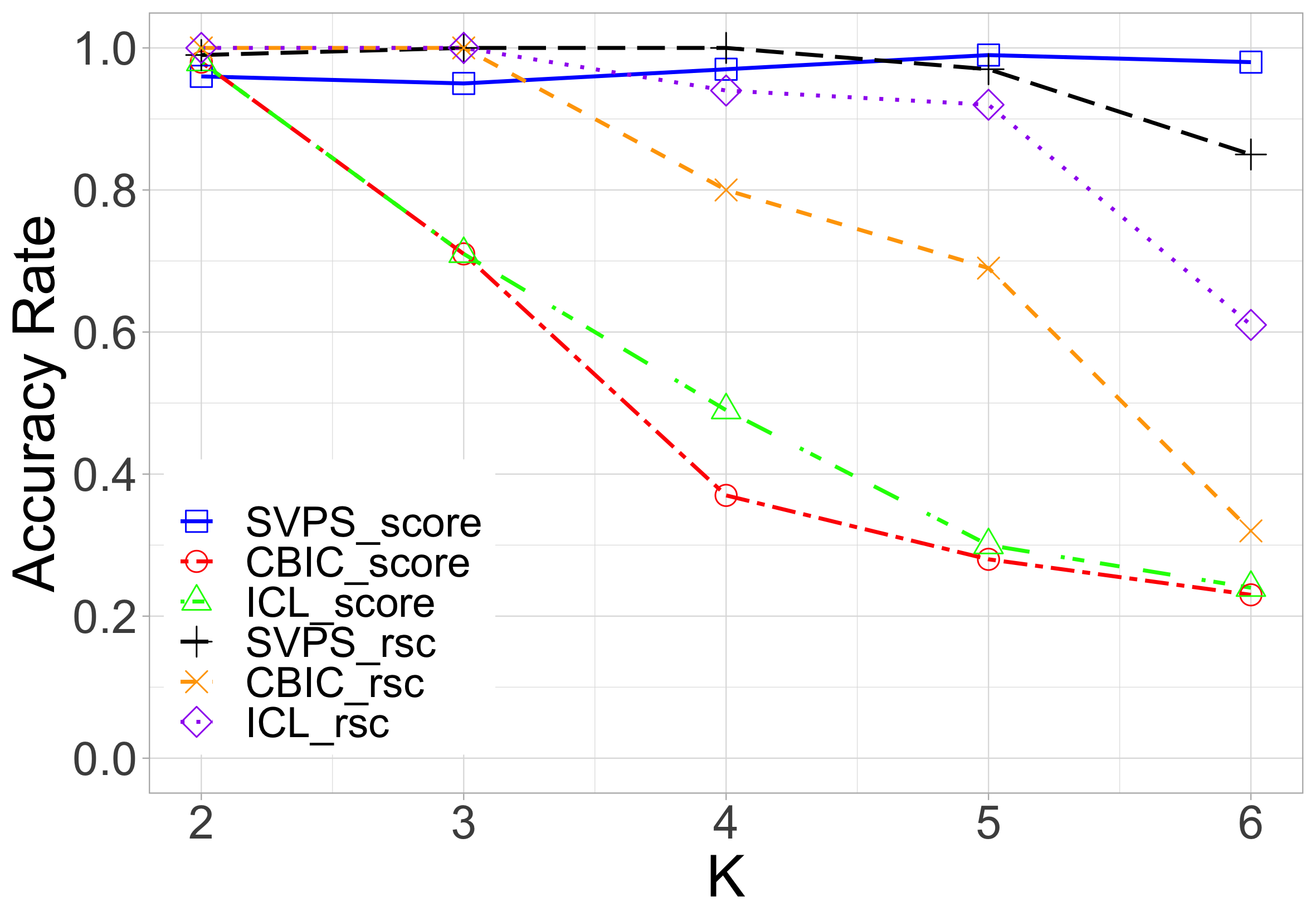}
  \caption{$\rho = 0.04$, $r=4$; $\epsilon = 0.05$}
\end{subfigure}%
\begin{subfigure}{.32\textwidth}
  \centering
  \includegraphics[width=1.0\linewidth, height=0.7\linewidth]{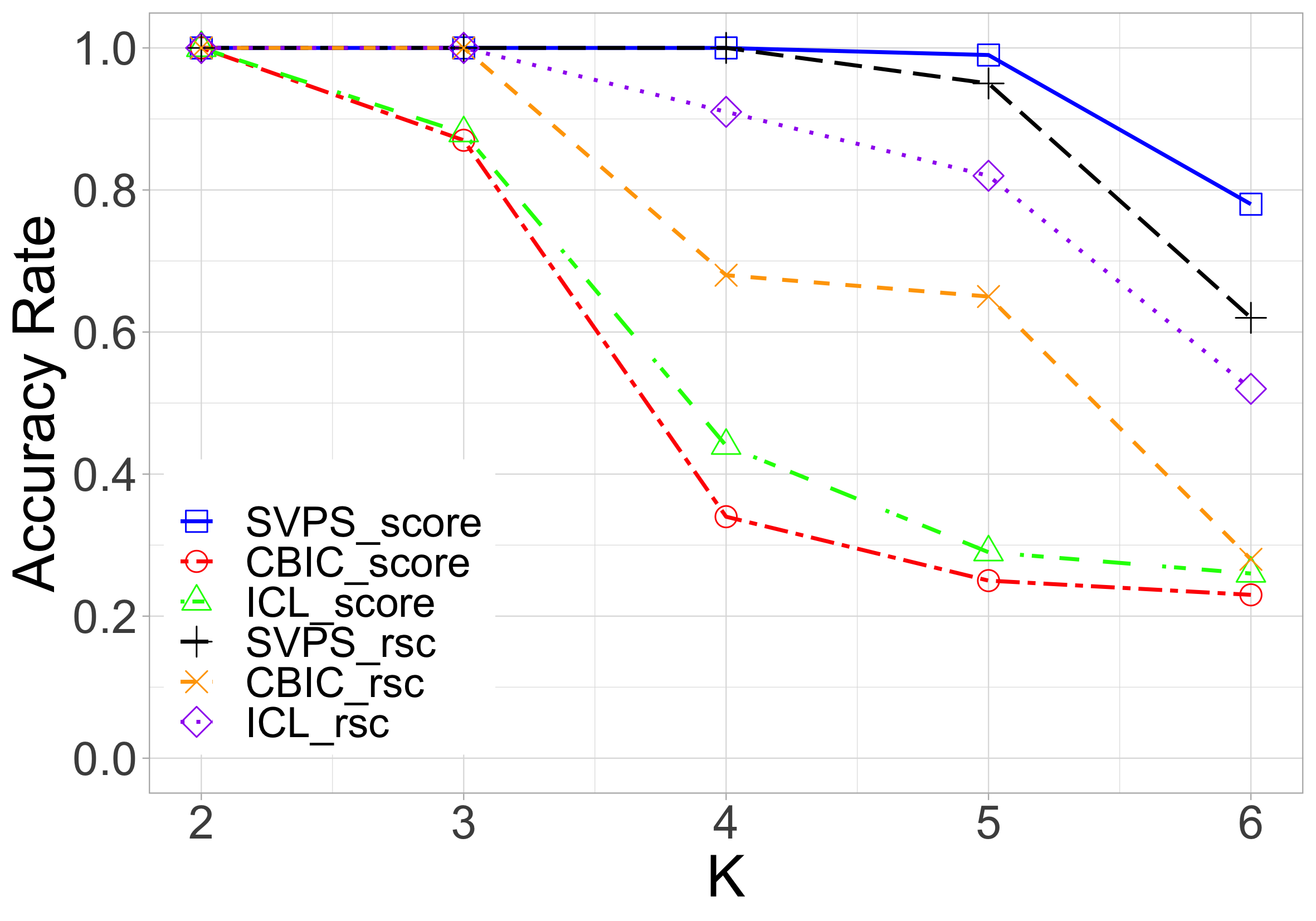}
  \caption{$\rho = 0.06$, $r=3$; $\epsilon = 0.05$}
\end{subfigure}
\begin{subfigure}{.32\textwidth}
  \centering
  \includegraphics[width = 1.0\linewidth, height=0.7\linewidth]{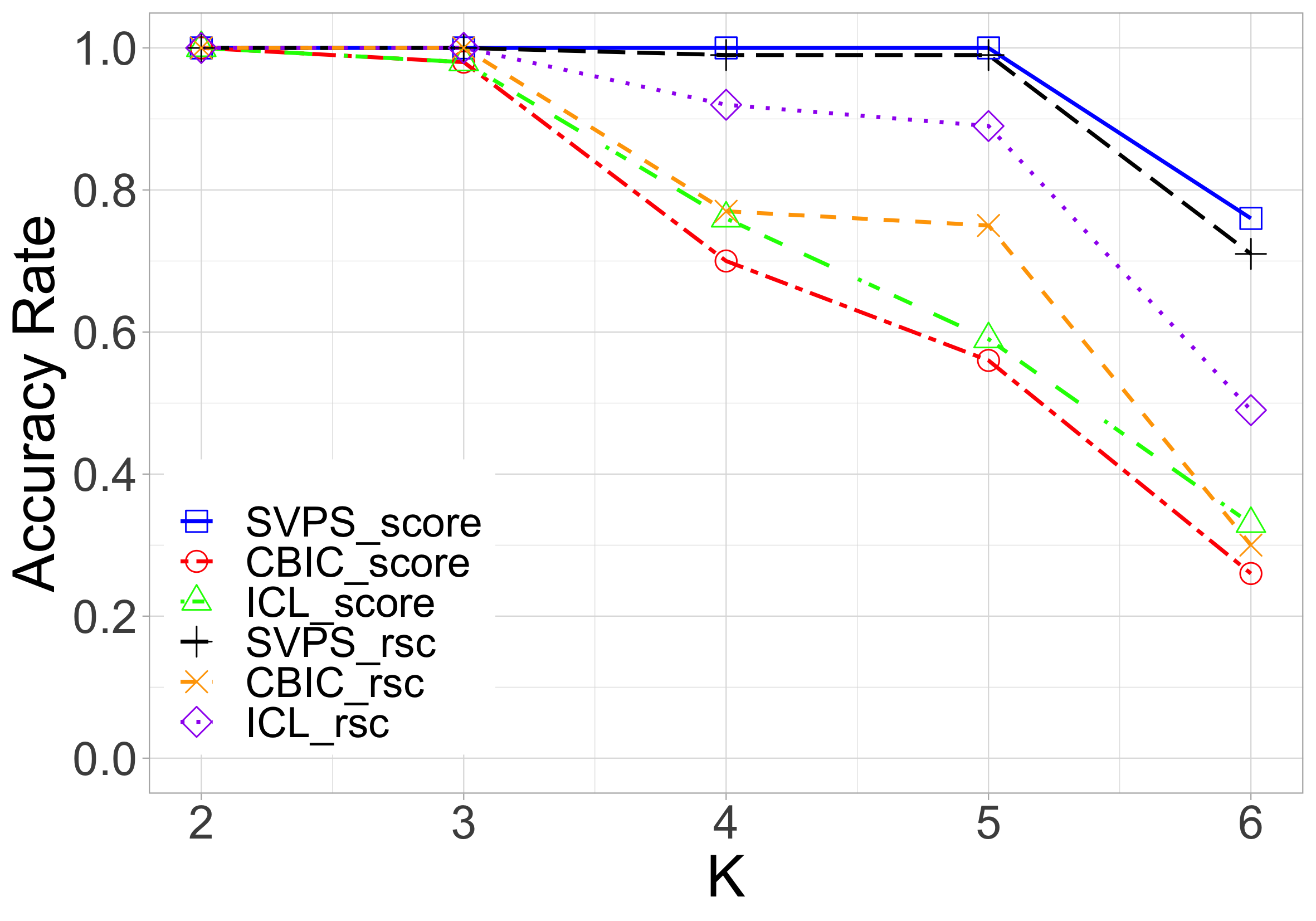}
  \caption{$\rho = 0.12$, $r=2$; $\epsilon = 0.02$}
\end{subfigure}
\hspace{0.003\textwidth}
\raisebox{3.8\height}{\rotatebox[origin=c]{90}{\scalebox{0.8}{Binomial}}}
\end{minipage}%

\begin{minipage}{1\textwidth}
\begin{subfigure}{.32\textwidth}
  \centering
  \includegraphics[width=1.0\linewidth, height=0.7\linewidth]{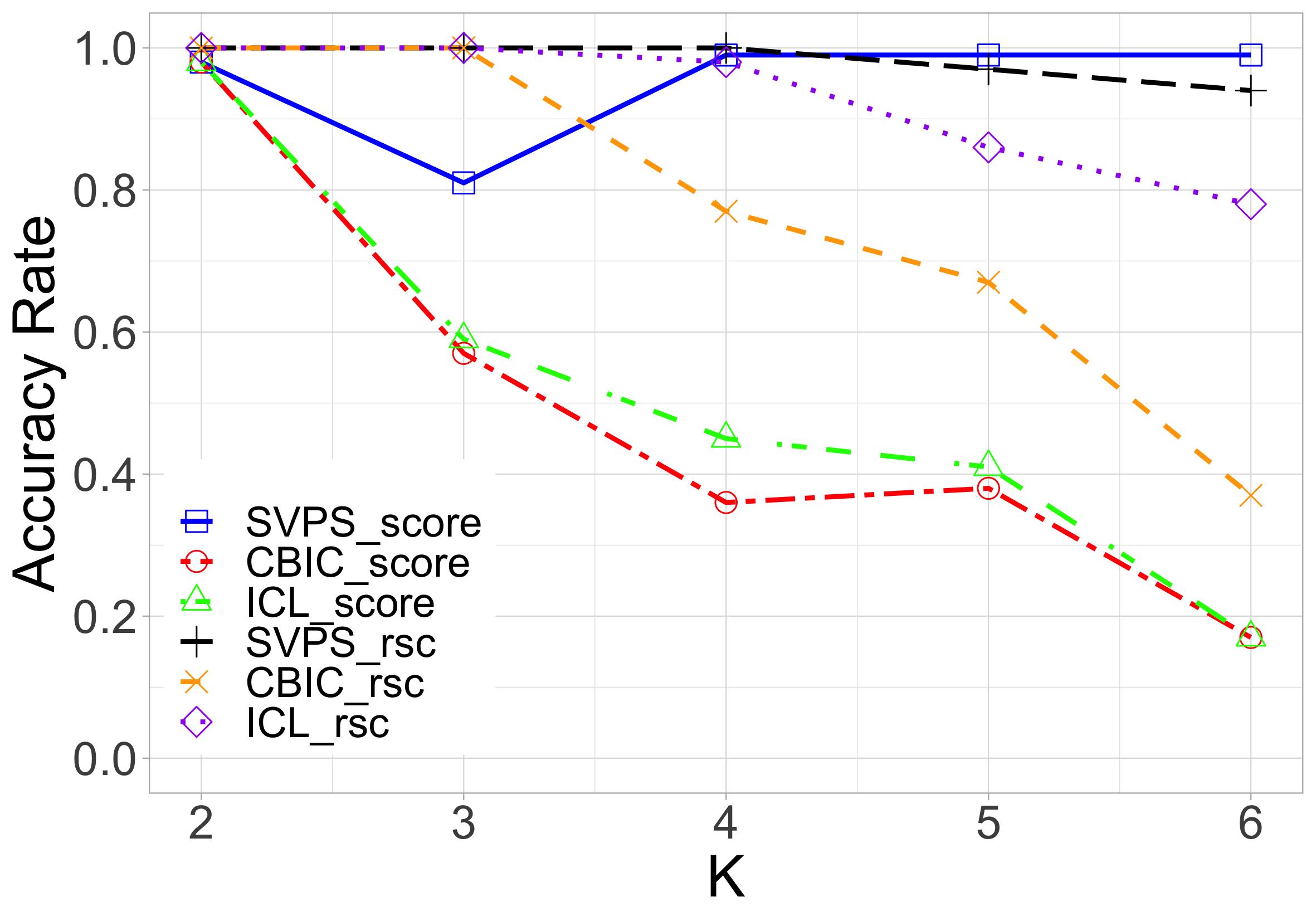}
  \caption{$\rho = 0.04$, $r=4$; $\epsilon = 0.1$}
\end{subfigure}%
\begin{subfigure}{.32\textwidth}
  \centering
  \includegraphics[width=1.0\linewidth, height=0.7\linewidth]{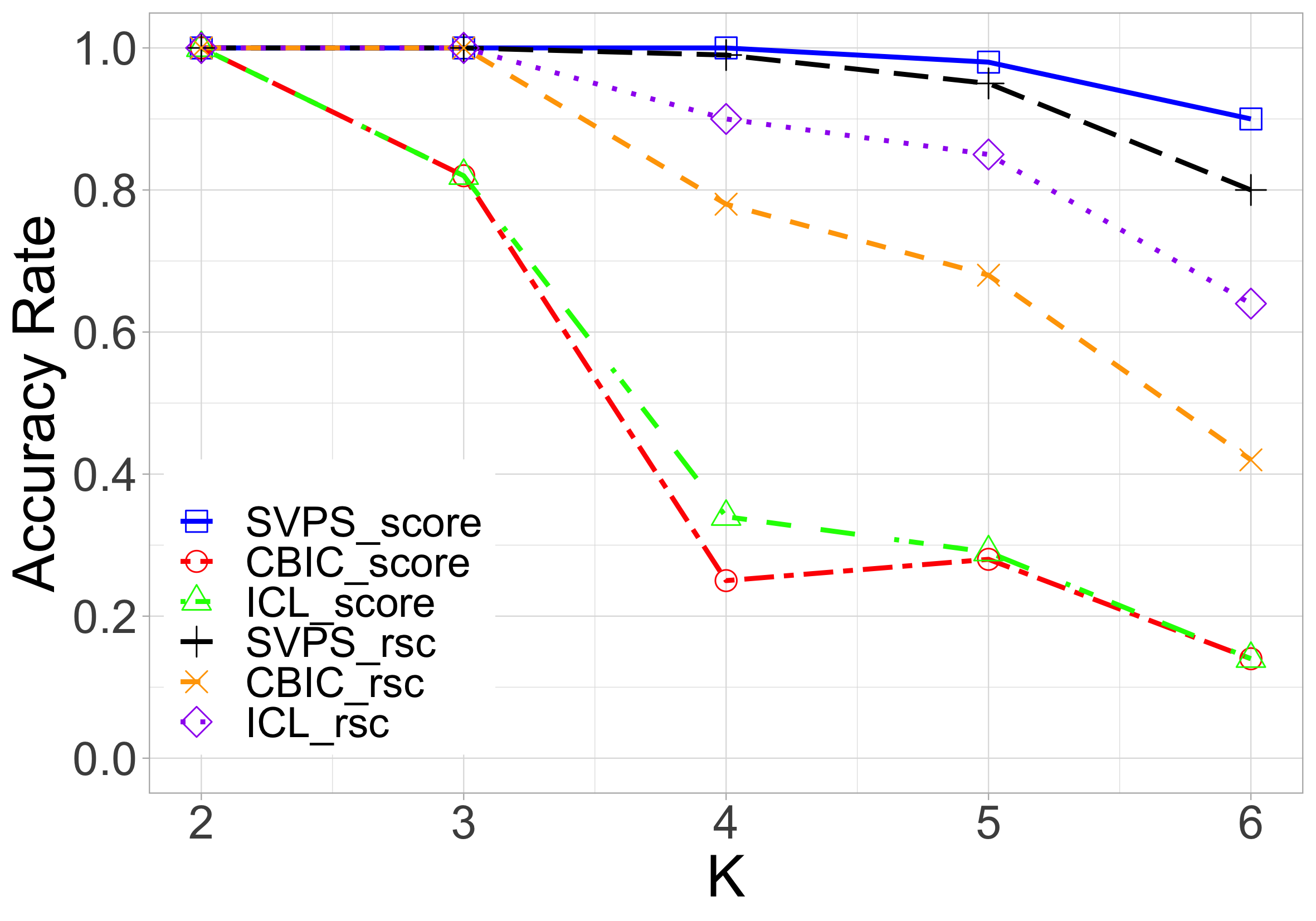}
  \caption{$\rho = 0.06$, $r=3$; $\epsilon = 0.1$}
\end{subfigure}
\begin{subfigure}{.32\textwidth}
  \centering
  \includegraphics[width=1.0\linewidth, height=0.7\linewidth]{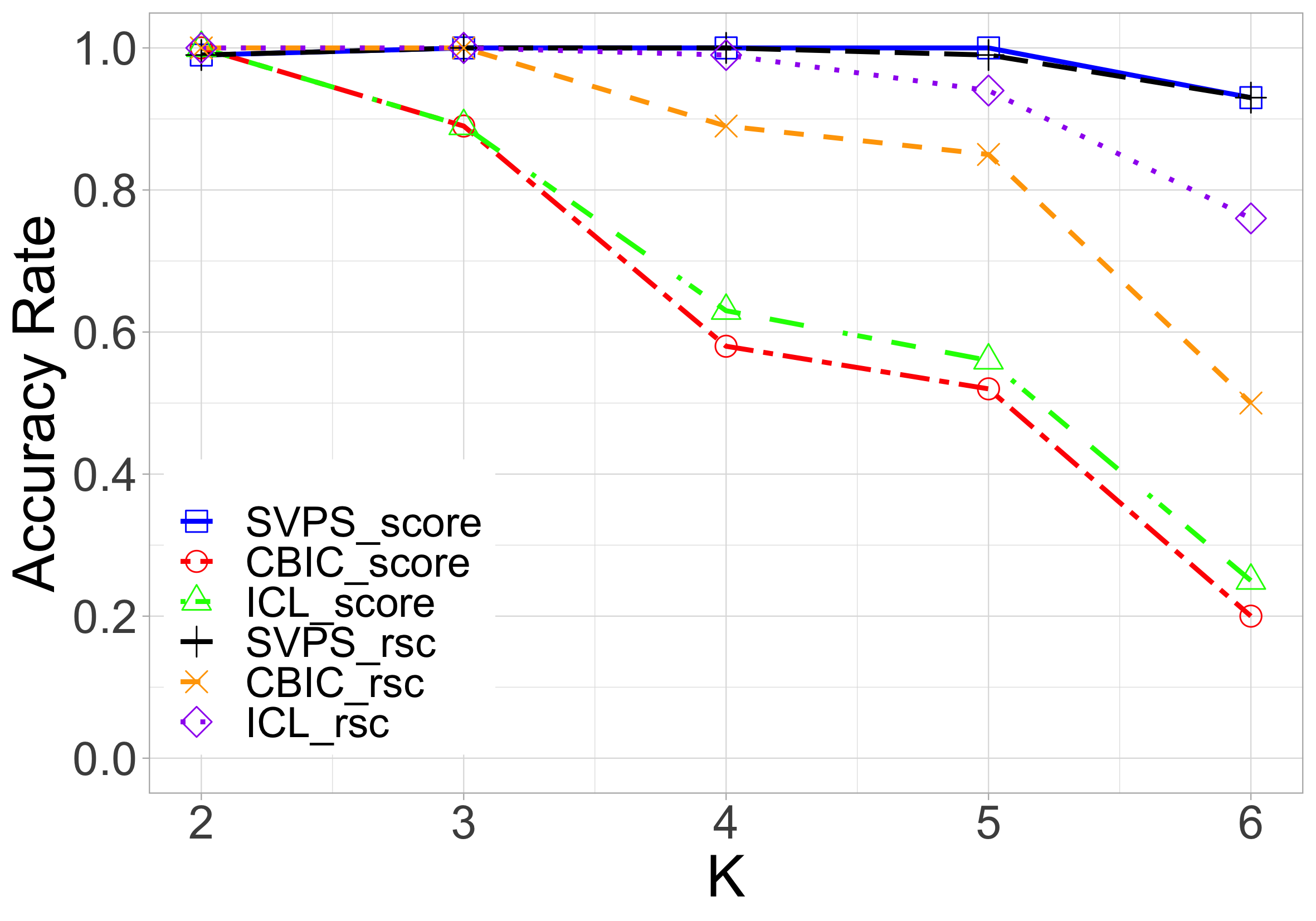}
  \caption{$\rho = 0.12$, $r=2$; $\epsilon = 0.1$}
\end{subfigure}
\hspace{0.001\textwidth}
\raisebox{2.1\height}{\rotatebox[origin=c]{90}{\scalebox{0.8}{Negative Binomial}}}
\end{minipage}%

\caption{Accuracy rate comparison between SVPS, CBIC and ICL in different simulations. The top, middle and bottom rows correspond to the Poisson, binomial and negative binomial distribution, respectively. }
\label{fig:simulation}
\end{figure*}

\subsubsection{Comparison with Score Based Methods}

In the simulated experiments, SVPS is compared with two score based methods: corrected BIC (CBIC) in \citet{hu2020corrected} and integrated classification likelihood (ICL) method in \citet{daudin2008mixture}, both of which are adapted to the weighted DCSBM with explicit likelihood function. To be specific, for each $m=1, 2, ...$, we calculate the following two scores
\begin{align*}
    \text{CBIC}(m) = \log f(\mtx{A} \vert \widehat{\mtx{M}}^{(m)}     ) - \left[ \lambda n \log m + \frac{m(m+1)}{2} \log n \right]
\end{align*}
and
\begin{align*}
    \text{ICL}(m) = & \log f (\mtx{A} \vert \widehat{\mtx{M}}^{(m)}     ) \\
    & - \left[ \sum_{k=1}^m \hat{n}_k \log \left( \frac{n}{\hat{n}_k} \right) + \frac{m(m+2)}{2} \log n \right],
\end{align*}
where $\lambda $ is a tuning parameter, $\widehat{\mtx{M}}^{(m)}$ is the estimated mean matrix, $\hat{n}_k$'s are the estimated block sizes according to $\widehat{\mtx{M}}^{(m)} $, and $f(\mtx{A} \vert \widehat{\mtx{M}}^{(m)})$ is the likelihood function based on the actual generating distribution of the network. Throughout all our experiments, we fix $\lambda = 1$ which is the common choice used in \citet{hu2020corrected}.
Note that although ICL is derived under the SBM, we extend its usage to the weighted DCSBM, similar to the experiments in \citet{hu2020corrected}. 
Ideally, $\widehat{\mtx{M}}^{(m)}$ should be based on the maximum likelihood. However, for the sake of fair comparison with SVPS, we only consider estimating the mean adjacency matrix by spectral clustering as discussed in Section \ref{sec:dcsbm_fitting}.

Next, we compare SVPS with CBIC and ICL in the three aforementioned simulations, experimenting with different combinations of $(\rho, r)$ and $K=2, 3, \ldots, 6$, where the spectral clustering method is either SCORE or RSC. In each simulation setup, we generate $100$ independent networks and record the percentage of correctly estimating $K$ of each method as the accuracy rate for comparison.

The plots of the empirical accuracy rates in the three simulations are shown in Figure \ref{fig:simulation}, with figure captions indicating $(\rho, r)$ and the choice of $\epsilon$ which determines the threshold in SVPS. As we can see, SVPS in general performs much better than the score based methods in these simulations, especially when using SCORE for spectral clustering. This result may seem surprising given that the score based methods utilize the likelihood information of the actual distribution, whereas SVPS does not even use the correct variance function. One possible explanation for the underperformance of the score based methods is that we only fit the models with spectral methods instead of maximum likelihood estimation with EM algorithms. Again, we choose spectral methods for fair comparison, since there is no likelihood as a guidance to fit the weighted DCSBM for SVPS. Another point to note is that in the first column of Figure \ref{fig:simulation}, when $K=3$, we observe a decrease in performance of SVPS with SCORE; however, the underlying reason remains unclear.

\subsection{Les Mis\'{e}rables Network}

In this subsection, we study the \textit{Les Mis\'{e}rables} weighted network compiled in \citet{knuth1993stanford}, also analyzed in the literature of network analysis, see, e.g., \citet{newman2004finding, ball2011efficient, newman2016estimating}. In this network, any two characters (nodes) are connected by a weighted edge representing the number of co-occurrences between the pair in the same chapter of the book. The estimated number of communities by some model based approach in \citet{newman2016estimating} is $6$, which may be reasonable since corresponding major subplots are also identified in their analysis.

With this co-occurrence network, we investigate the difference between the estimated numbers of communities when the network is treated as unweighted and weighted. In the unweighted case, by replacing any positive weight with $1$, we apply CBIC and ICL based on the standard DCSBM, as well as two popular methods for unweighted networks: the Bethe Hessian matrix (BH) proposed in \citet{le2015estimating} and stepwise goodness-of-fit (StGoF) proposed in \citet{jin2022optimal} to the unweighted adjacency matrix to obtain the estimated number of communities $\widehat{K}$. We use SCORE for node clustering in CBIC, ICL and StGoF.  Note that the original StGoF fails to stop in this case, so we choose $\widehat{K}$ as the number of communities which corresponds to the smallest test statistic, as suggested in \citet{jin2022optimal}. The results of rank selection by these methods are listed in Table \ref{tab:lesmis_unweighted}, which are either $3$ or $4$, much smaller than $6$ as suggested in \citet{newman2016estimating}. We apply RSC and SCORE with $K=3$ to the unweighted network, and the clustering results are shown in Figure \ref{fig:lesmis_unweighted_3}.

\begin{table}[h!]
\setlength{\tabcolsep}{10pt}
\renewcommand{\arraystretch}{1.5}
    \centering
    \caption{Estimated $K$ in the unweighted adjacency matrix of Les Mis\'{e}rables by different methods.}
    \begin{tabular}{lcccc}
    \hline
    & BH & CBIC & ICL & StGoF
    \\ \hline
    $\widehat{K}$  & 4 & 3 & 3 & $3$
    \\
    \hline
    \end{tabular}
    \label{tab:lesmis_unweighted}
\end{table}

\begin{figure}[t!]
\begin{subfigure}{.5\textwidth}
  \centering
  \includegraphics[width=0.5\linewidth, height=0.5\linewidth]{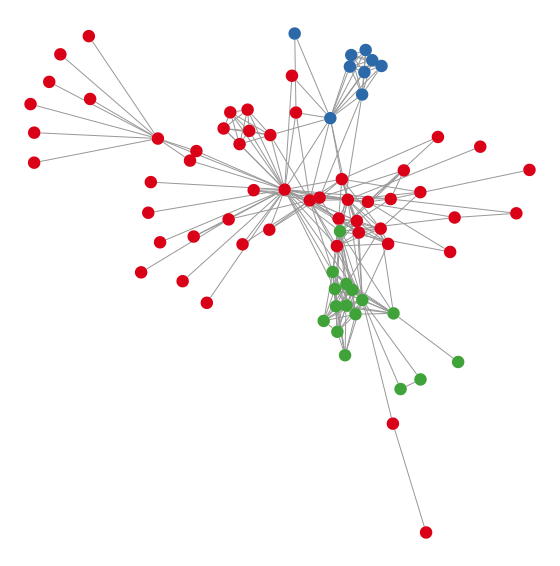}
  \caption{Unweighted with $K=3$ and RSC}
\end{subfigure}%
\begin{subfigure}{.5\textwidth}
  \centering
  \includegraphics[width=0.5\linewidth, height=0.5\linewidth]{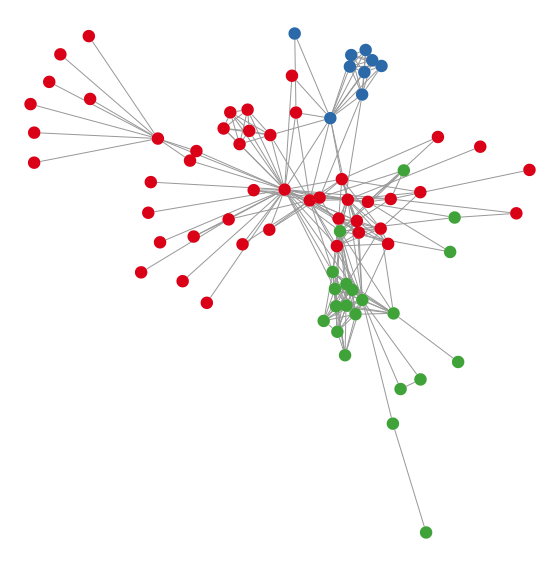}
  \caption{Unweighted with $K=3$ and SCORE}
\end{subfigure}%
\caption{Estimated clusters of nodes in Les Mis\'{e}rables by applying RSC and SCORE to the unweighted adjacency with $K=3$. }
\label{fig:lesmis_unweighted_3}
\end{figure}

Coming back to the original weighted network, we compare the results of rank selection by SVPS, CBIC and ICL using both SCORE and RSC for spectral clustering. CBIC and ICL are based on the likelihood of the Poisson DCSBM.
For the implementation of SVPS, we first preprocess the weighted network into a regularized adjacency matrix $\mtx{A}_{\tau} = \mtx{A}+\tau \mtx{J}_n$, where $\mtx{J}_n$ is the $n\times n$ all-one matrix and $\tau=0.05, 0.1, 0.25, 0.5$. The threshold for sequential testing in SVPS is set to 2.05. The comparison of the estimated $K$ is summarized in Table \ref{tab:lesmis_weighted}, which shows some consistency among these methods. Additionally, Figure \ref{fig:lesmis_weighted_567} displays the estimated clusters by applying SCORE and RSC to the weighted network with $K=5, 6, 7$. It seems that RSC yields more interpretable results than SCORE.

\begin{table}[h!]
\renewcommand{\arraystretch}{1.5}
    \centering
    \caption{Estimated $K$ in Les Mis\'{e}rables by applying SVPS to the regularized weighted adjacency with different $\tau$ values and CBIC, ICL to the original weighted adjacency with Poisson likelihood. The rows correspond to SCORE and RSC as the clustering method.}
    \begin{tabular}{lcccccc}
    \hline
    & \makecell{SVPS \\ ($\tau=0.05$)}  & \makecell{SVPS \\ ($\tau=0.1$)} & \makecell{SVPS \\ ($\tau=0.25$)} & \makecell{SVPS \\ ($\tau=0.5$)} & CBIC & ICL 
    \\ \hline
    SCORE  & 7 & 6 & 7 & 6 & 7 & 7
    \\
    RSC  & 7 & 7 & 6 & 6 & 7 & 5
    \\
    \hline
    \end{tabular}
    \label{tab:lesmis_weighted}
\end{table}

\begin{figure}[t!]
\begin{subfigure}{.5\textwidth}
  \centering
  \includegraphics[width=0.5\linewidth, height=0.5\linewidth]{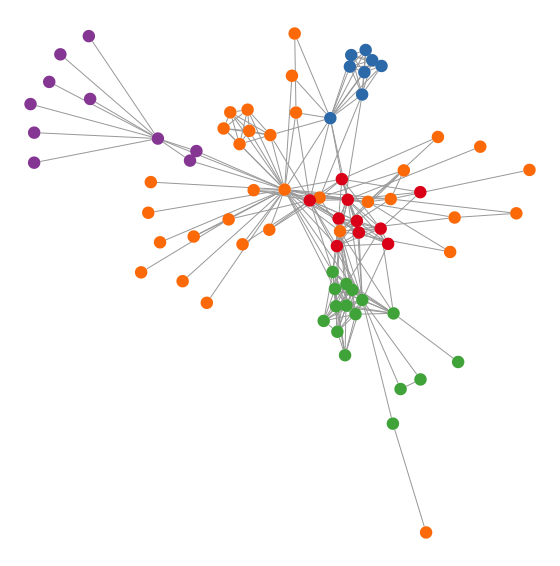}
  \caption{Weighted with $K=5$ and RSC}
\end{subfigure}
\begin{subfigure}{.5\textwidth}
  \centering
  \includegraphics[width=0.5\linewidth, height=0.5\linewidth]{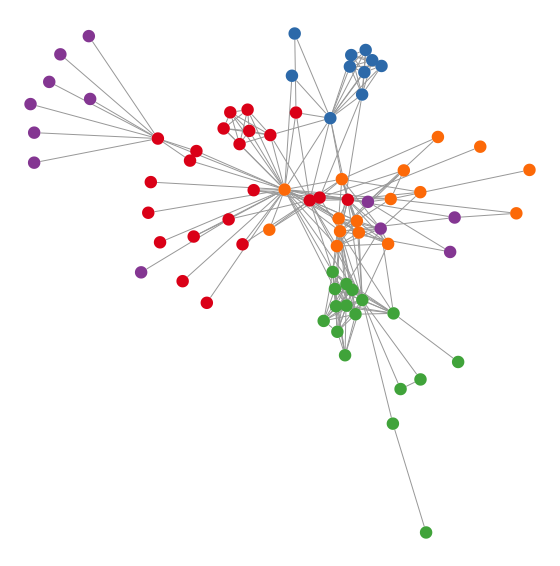}
  \caption{Weighted with $K=5$ and SCORE}
\end{subfigure}
\begin{subfigure}{.5\textwidth}
  \centering
  \includegraphics[width=0.5\linewidth, height=0.5\linewidth]{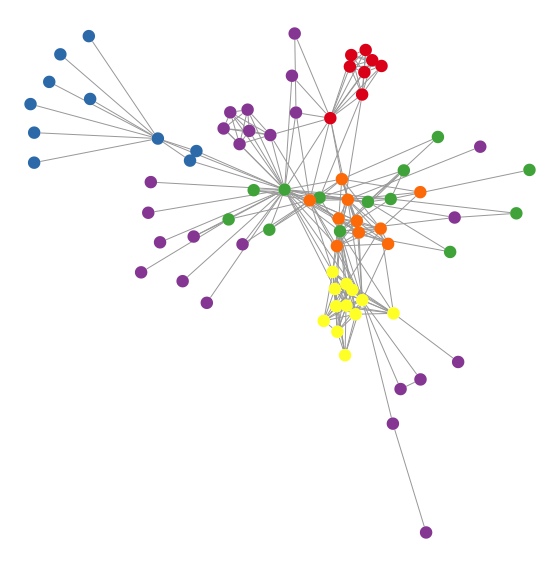}
  \caption{Weighted with $K=6$ and RSC}
\end{subfigure}
\begin{subfigure}{.5\textwidth}
  \centering
  \includegraphics[width=0.5\linewidth, height=0.5\linewidth]{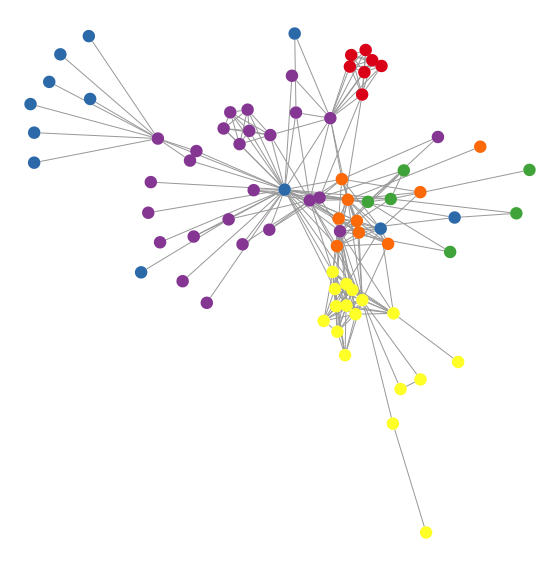}
  \caption{Weighted with $K=6$ and SCORE }
\end{subfigure}
\begin{subfigure}{.5\textwidth}
  \centering
  \includegraphics[width=0.5\linewidth, height=0.5\linewidth]{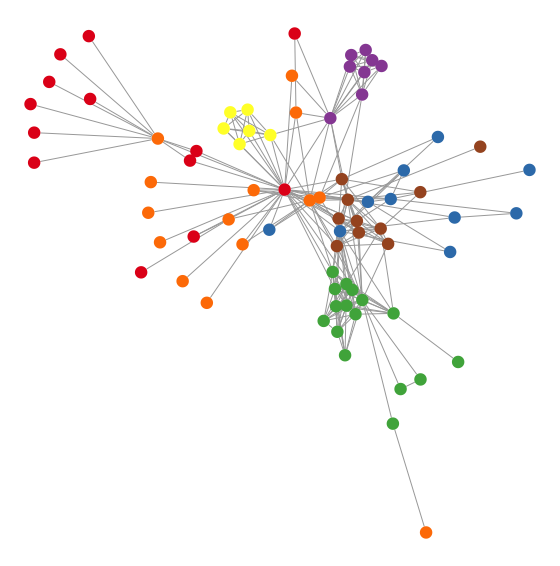}
  \caption{Weighted with $K=7$ and RSC}
\end{subfigure}
\begin{subfigure}{.5\textwidth}
  \centering
  \includegraphics[width=0.5\linewidth, height=0.5\linewidth]{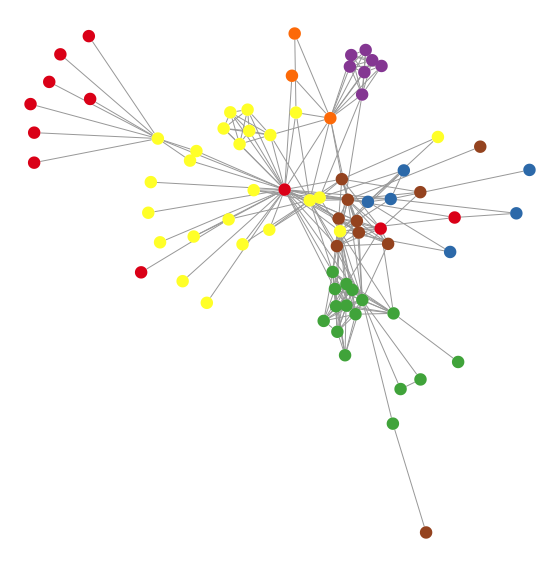}
  \caption{Weighted with $K=7$ and SCORE}
\end{subfigure}
\caption{Estimated clusters in Les Mis\'{e}rables by applying RSC or SCORE to the weighted adjacency with $K=5, 6, 7$.}
\label{fig:lesmis_weighted_567}
\end{figure}

\section{DISCUSSION}
\label{sec:discussion}

This article studies a method for selecting the number of communities for weighted networks. First, we propose a generic weighted DCSBM, where the mean adjacency matrix is modeled as a DCSBM, and the variance profile is linked to the mean adjacency matrix through a variance function, without likelihood imposed. Next, we introduce a sequential testing approach for selecting the number of communities. In each step, the mean structure is fitted with some clustering method, and the variance profile matrix can be estimated through the variance function. A key component in constructing the test statistic is the scaling of the variance profile matrix, which is helpful in normalizing the adjacency matrix. The test statistic, obtained from the normalized adjacency matrix, is then used to determine the number of communities.

In theory, we show the consistency of our proposed procedure in selecting the number of communities under mild conditions on the weighted DCSBM. In particular, the network is allowed to be sparse. The consistency results include analysis of both the null and the underfitting cases. Additionally, the Nonsplitting Property studied in \citet{jin2022optimal} can be extended to the weighted DCSBM, which plays an essential role in our theoretical analysis. 

For future work, a number of questions remain, both theoretically and methodologically. In theory, we assume that the weighted DCSBM satisfies some balanced conditions for analytical convenience, but this condition is very likely to be further relaxed. 
Our theoretical results are established for a known variance function $\nu(\cdot)$, and it would be interesting to have analogous results for an estimated variance function $\hat{\nu}(\cdot)$. Moreover, it would be valuable to study the asymptotic null distribution of $T_{n, m}$, based on which the p-value can be calculated at each step. However, we are unable to derive distributional results for $T_{n, m}$ based on the tools used in this paper. 

A key assumption for our method is that the relationship between the variance profile matrix and the mean adjacency matrix can be characterized by a variance function. It would be interesting to see whether this assumption can be relaxed. Note that the estimation of the variance profile is for the purpose of matrix scaling so that the normalization matrix $\mtx{\Psi}$ can be estimated under the null case. It would be valuable to investigate whether $\mtx{\Psi}$ can be consistently estimated even if the variance profile matrix cannot be consistently estimated. We leave these aforementioned questions for future study.

\section*{Acknowledgment}
Y. Liu and X. Li are partially supported by the NSF via the Career Award DMS-1848575. X. Li would like to thank Tracy Ke, Can Le, Haoran Li, Kaizheng Wang, and Ke Wang for their helpful discussions.

\bibliographystyle{plainnat}
    \bibliography{ref}

\appendix

\newpage

\section*{ORGANIZATION OF THE SUPPLEMENTARY MATERIAL}
The supplementary material consists of four appendices. Appendix \ref{sec:lemmas} contains the supporting lemmas of the main results. Appendix \ref{sec:proof_main} presents the proofs of Theorem \ref{thm:null} and Theorem \ref{thm:under-fitting}. Appendix \ref{sec:proof_nsp} is the proof of Lemma \ref{thm:nsp}, which proves the Nonsplitting Property of SCORE under the weighted DCSBM. Appendix \ref{sec:additional_exp} presents additional synthetic simulations with larger values of $K$.

\section{SUPPORTING LEMMAS}
\label{sec:lemmas}
We first cite two results from \citet{landa2022scaling} regarding the sensitivity analysis of matrix scaling. Here we only state the results for symmetric matrix scaling with row sums equal to $n$.
\begin{lemma}[\citet{landa2022scaling}]
\label{lem:landa_2}
Let $\mtx{A}$ be an $n \times n$ symmetric matrix with positive entries. Then, there exists a unique positive vector $\vct{x} \in \mathbb{R}^n$ satisfying
\[
x_i \left(\sum_{j=1}^n A_{ij}x_j \right) = 1 , \quad i=1, \ldots, n.
\]
Furthermore, denote $a_{\max} = \max_{1 \leq i \leq j \leq n} A_{ij}$ and $a_{\min} = \min_{1 \leq i \leq j \leq n} A_{ij}$. Then,
\[
\frac{1}{\sqrt{n}}\frac{\sqrt{a_{\min}}}{a_{\max}} \leq x_i \leq \frac{1}{\sqrt{n}}\frac{\sqrt{a_{\max}}}{a_{\min}}, \quad i=1, \ldots, n.
\]
\end{lemma}

\begin{lemma}[\citet{landa2022scaling}]
\label{lem:landa_9}
Let $\widetilde{\mtx{A}}$ be a symmetric matrix with positive entries. Denote $\tilde{a}_{\max} = \max_{1\leq i \leq j \leq n} \widetilde{A}_{ij}$ and $\tilde{a}_{\min} = \min_{1\leq i \leq j \leq n} \widetilde{A}_{ij}$. Suppose there is a constant $\epsilon \in (0, 1)$ and a positive vector $\vct{x} \in \mathbb{R}^n$, such that
\[
\max_{1 \leq i \leq n} \left\vert \sum_{j=1}^n x_i \widetilde{A}_{ij} x_j - 1 \right\vert \leq \epsilon.
\]
Denote $x_{\min} = \min_{1 \leq i \leq n} x_i$. Then, there exists a positive vector $\tilde{\vct{x}} \in \mathbb{R}^n$ such that
\[
\sum_{j=1}^n \tilde{x}_i \widetilde{A}_{ij} \tilde{x}_j = 1, \quad i=1, \ldots, n,
\]
and
\[
\max_{1 \leq i \leq n} \left\vert \frac{\tilde{x}_i}{x_i} - 1\right\vert\leq \frac{\epsilon}{1-\epsilon} + 4 \epsilon \cdot \frac{n^{3/2}}{x_{\min}^3}\cdot\frac{\sqrt{\tilde{a}_{\max}}}{\tilde{a}_{\min}^2 }.
\]
As a consequence, if $\vct{x}$ satisfies the equations in Lemma \ref{lem:landa_2}, there holds
\[
\max_{1 \leq i \leq n} \left\vert \frac{\tilde{x}_i}{x_i} - 1\right\vert\leq \frac{\epsilon}{1-\epsilon} + 4 \epsilon \cdot \frac{a_{\max}^3}{a_{\min}^{3/2}}\cdot\frac{\tilde{a}_{\max}^{1/2}}{\tilde{a}_{\min}^2 }.
\]
\end{lemma}

The following lemma provides a tail bound of the sum of independent random variables satisfying the Bernstein condition.

\begin{lemma}[Bernstein's Inequality, Corollary 2.11 of \citet{boucheron2013concentration}]
\label{lem:bernstein_inequality}
Let $X_1,\ldots,X_n$ be independent real-valued random variables. Assume that there exist positive numbers $v$ and $c$ such that $\sum_{i=1}^n \E [X_i^2]\leq v$ and
\begin{equation*}
    \sum_{i=1}^n \E \left[|X_i|^q \right] \leq \frac{q!}{2} c^{q-2} v  \quad \text{for all intergers } q\geq 3.
\end{equation*}
Then, for all $t > 0$,
\[
\P \left(\sum_{i=1}^n \left(X_i - \E \left[ X_i \right] \right) \geq t \right) \leq \exp\left( -\frac{t^2}{2(v+ct)} \right).
\]
\end{lemma}


The following lemma given in \citet{latala2018dimension} bounds the operator norm of heterogeneous random matrices.
\begin{lemma}[Remark 4.12 of \citet{latala2018dimension}]
\label{lem:E_spec_norm}
Suppose that $\mtx{X}$ is a random matrix with independent and centered upper-triangular entries. Define the quantities
\[
 \sigma_\infty  \coloneqq  \underset{i}{\max }\sqrt{\sum_{j}\E \left[X^2_{ij} \right]} , \quad \sigma_\infty^*  \coloneqq \underset{i,j}{\max } \left\| X_{ij} \right\|_{\infty} .
\]
Then, for every $0\leq \epsilon \leq 1$ and $t\geq 0$, we have
\[
\mathbb{P} \left( \| \mtx{X} \| \geq 2(1+\epsilon) \sigma_{\infty} + t\right) \leq n \exp\left(- \frac{\epsilon t^2}{C \sigma_{\infty}^{*2}} \right) ,
\]
where $C$ is a universal constant.
\end{lemma}

\begin{lemma}
\label{lem:spectral_lowerbound}
Under Assumption \ref{ass:DCSBM}, there holds $\left|\lambda_K\left( \mtx{M}\right) \right| \geq c_0 \theta_{\min}^2   n$. 
\end{lemma}

\begin{proof}
Recall that $\mtx{M} = \mtx{\Theta} \mtx{\Pi} \mtx{B} \mtx{\Pi}^\top \mtx{\Theta}$, where $\mtx{\Theta} \mtx{\Pi} = \left[ \vec{\theta}_1, \vec{\theta}_2, \ldots, \vec{\theta}_K \right]$. Denote $\sigma_k \left( \cdot \right)$ as the $k$-th singular value of a matrix. It is easy to observe
\[
\sigma_K \left(\mtx{\Theta} \mtx{\Pi} \right) = \underset{k}{\min} \| \vec{\theta}_k \|_2 \geq \theta_{\min} \sqrt{n} .
\]
Then by the assumption \eqref{eq:B_entry_ass}, we have
\[
|\lambda_K(\mtx{M})| \geq \left| \lambda_K(\mtx{B}) \right| \cdot \sigma_K\left(\mtx{\Theta} \mtx{\Pi} \right)^2  \geq c_0  \theta_{\min}^2 n.
\]
\end{proof}

\begin{lemma}
\label{lem:psi_bound}
Under Assumption \ref{ass:DCSBM}, the scaling factors satisfy 
\[
\left. 1\middle/\left(C \theta_{\min} \sqrt{n}\right) \right. \leq \psi_i \leq \left. C \middle/\left(\theta_{\min} \sqrt{n}\right) \right., \quad i=1, \ldots, n
\]
for some constant $C$ that only relies on $c_0$.
\end{lemma}

\begin{proof}
In light of Assumption \ref{ass:DCSBM}, this result is implied by Lemma \ref{lem:landa_2} directly.
\end{proof}

\begin{lemma}
\label{lem:E_norm}
For any fixed $c_0>0$ in Assumption \ref{ass:DCSBM}, there holds
\[
\left\| \mtx{\Psi}^{\frac{1}{2}}(\mtx{A}-\mtx{M})\mtx{\Psi}^{\frac{1}{2}} \right\| \leq 2+o_P(1) \quad \text{as~} n \rightarrow \infty. 
\]
\end{lemma}

\begin{proof}
Denote $\mtx{E} \coloneqq \mtx{\Psi}^{\frac{1}{2}}(\mtx{A}-\mtx{M})\mtx{\Psi}^{\frac{1}{2}}$ whose entries are $E_{ij} = (A_{ij} - M_{ij}) \psi_i^{\frac{1}{2}} \psi_j^{\frac{1}{2}}$. By \eqref{eq:pop_scaling}, we have
\begin{equation}
\label{eq:double_stochastic}
\sum_{j=1}^n \E \left[E_{ij}^2 \right] =1, \quad \forall~i=1,\ldots,n.
\end{equation}
Define 
\[
\widehat{E}_{ij} = E_{ij} \mathbf{1}_{\{|E_{ij}| < 1/(\log n)\}} \quad \text{and} \quad \widetilde{E}_{ij} = \widehat{E}_{ij} - \E \left[\widehat{E}_{ij}\right].
\]
The resulting matrices are denoted as $\widehat{\mtx{E}}$ and $\widetilde{\mtx{E}}$, respectively. Our proof strategy is to first bound $\| \widetilde{\mtx{E}} \| $, and then $\| \widehat{\mtx{E}} \| $, and finally $\| \mtx{E} \| $.

Notice that $\widetilde{\mtx{E}}$ has independent and mean-centered entries for $i\geq j$. Since 
\[
\left|\widehat{E}_{ij} \right| = \left|E_{ij}\right| \mathbf{1}_{\{|E_{ij}| < 1/(\log n)\}} \leq \left|E_{ij}\right|,
\]
almost surely we have
\begin{align*}
\E \left[\widetilde{E}^2_{ij}\right] = \Var \left( \widetilde{E}_{ij} \right) = \Var \left( \widehat{E}_{ij} \right) \leq \E \left[\widehat{E}_{ij}^2\right] \leq \E \left[E_{ij}^2\right].
\end{align*}
This implies that
\[
\sum_{j=1}^n \E \left[\widetilde{E}_{ij}^2 \right] \leq 1, \quad \forall~ i=1,\ldots,n.
\]
Also, since $\| \widehat{E}_{ij} \|_{\infty}\leq 1/(\log n)$, we have $\vert \E [ \widehat{E}_{ij} ] \vert \leq 1/(\log n)$, and thus 
\[
\| \widetilde{E}_{ij} \|_{\infty} \leq \| \widehat{E}_{ij} \|_{\infty} + \vert \E \widehat{E}_{ij} \vert \leq 2/(\log n).
\]
Define the quantities
\[
 \tilde{\sigma}  \coloneqq  \underset{i}{\max }\sqrt{\sum_{j}\E [\widetilde{E}^2_{ij}] } \leq 1, \quad \tilde{\sigma}_*  \coloneqq \underset{i,j}{\max } \| \widetilde{E}_{ij} \|_{\infty} \leq 2/(\log n).
\]
Then, apply Corollary \ref{lem:E_spec_norm} to $\widetilde{\mtx{E}}$. For any $t > 0$, there holds
\[
\mathbb{P} \left( \| \widetilde{\mtx{E}} \| \geq 2(1+\epsilon) \tilde{\sigma} + t\right) \leq n \exp\left(- \frac{\epsilon t^2}{C \tilde{\sigma}_*^2} \right).
\]
By letting $\epsilon = t = \log^{-1/4} n$, with the above inequalities, we have
\begin{equation}
\label{eq:E_tilde_spec_ori}
\| \widetilde{\mtx{E}} \| \leq 2+o_P(1) \quad \text{as~} n \rightarrow \infty. 
\end{equation}



Next, we give an upper bound of $\| \widehat{\mtx{E}} \|$. Since $\widetilde{\mtx{E}} = \widehat{\mtx{E}} - \E [ \widehat{\mtx{E}} ]$, we have $\| \widehat{\mtx{E}} \| \leq \| \widetilde{\mtx{E}} \| + \| \E [ \widehat{\mtx{E}} ] \|_F$.
Notice that
\[
\E \left[E_{ij} \mathbf{1}_{\{|E_{ij}| \geq 1/(\log n)\}}\right] = \E \left[E_{ij} - \widehat{E}_{ij}\right]  = -\E \left[\widehat{E}_{ij}\right].
\]
By Cauchy-Schwarz inequality, 
\begin{equation}
\label{eq:trunc_bias}
\left(\E \left[\widehat{E}_{ij}\right]\right)^2 = \left( \E \left[E_{ij} \mathbf{1}_{\{|E_{ij}| \geq 1/(\log n)\}}\right] \right)^2 \leq \E [E_{ij}^2 ]  \P \left( |E_{ij}| \geq 1/(\log n) \right).
\end{equation}
By Lemma \ref{lem:psi_bound}, there holds
\begin{align*}
\P \left(\left|E_{ij}\right| \geq 1/(\log n) \right)  &= \P \left(\left|A_{ij} - M_{ij}\right| \geq \frac{1}{\psi_i^{1/2} \psi_j^{1/2} \log n}  \right) 
\\
&\leq \P \left(\left|A_{ij} - M_{ij}\right| \geq \frac{\theta_{\min} \sqrt{n}}{C \log n}\right) .
\end{align*}
Then, by Assumption \ref{ass:DCSBM} and Lemma \ref{lem:bernstein_inequality}, we have
\begin{align}
\label{eq:trunc_effect}
\max_{1 \leq i, j \leq n} \P \left( \left|E_{ij}\right| \geq 1/(\log n) \right)  &\leq \max_{1 \leq i \leq j \leq n} 2 \exp\left( -\frac{ \left(\frac{\theta_{\min} \sqrt{n}}{C \log n} \right)^2 }{2\left( M_{ij}/c_0 + R \frac{\theta_{\min} \sqrt{n}}{C \log n} \right)} \right) \nonumber
\\
& \leq 2 \exp\left( - \frac{C \theta_{\min} n }{ \theta_{\min} \log^2 n +  \sqrt{n} \log n} \right) \nonumber
\\
& = O(n^{-10}),
\end{align}
where the last line is due to \eqref{eq:sparseness} in Assumption \ref{ass:DCSBM}. Then, \eqref{eq:double_stochastic} and \eqref{eq:trunc_bias} imply
\begin{align*}
\left\| \E [\widehat{\mtx{E}}] \right\|_F^2 &= \sum_{1 \leq i, j \leq n}\left(\E \widehat{E}_{ij}\right)^2 \leq \sum_{1 \leq i, j \leq n} \E [E_{ij}^2] \P \left( \left|E_{ij}\right| \geq 1/(\log n) \right) = O(n^{-9}).
\end{align*}
Then \eqref{eq:E_tilde_spec_ori} implies
\begin{equation}
\label{eq:E_hat}
\| \widehat{\mtx{E}} \| \leq \| \widetilde{\mtx{E}} \| + \| \E [ \widehat{\mtx{E}} ] \|_F \leq 2+o_P(1).
\end{equation}
Finally, note that 
\begin{align*}
\P \left(\widehat{\mtx{E}} \neq \mtx{E} \right) 
&= \P\left( \cup_{1 \leq i, j \leq n}\{|E_{ij}| \geq 1/(\log n)\} \right)
\\
&\leq \sum_{1 \leq i, j \leq n } \P (|E_{ij}| \geq 1/(\log n) )
\\
&\leq n^2 \max_{1 \leq i, j \leq n } \P (|E_{ij}| \geq 1/(\log n) ) = O(n^{-8}).
\end{align*}
Then \eqref{eq:E_hat} implies $\| \mtx{E} \| \leq 2+o_P(1)$.
\end{proof}

\begin{lemma}
\label{lem:noise_norm}
Under Assumption \ref{ass:DCSBM}, there holds
\[
\left\| \mtx{A}-\mtx{M} \right\| =  O_P(\theta_{\min} \sqrt{n}) \quad \text{as~} n \rightarrow \infty. 
\]
\end{lemma}
\begin{proof}
This is a straightforward corollary of Lemmas \ref{lem:psi_bound} and \ref{lem:E_norm}. 
\end{proof}

\begin{lemma}
\label{lem:degree_concentration}
Under Assumption \ref{ass:DCSBM}, we have
\[
\max_{1 \leq i \leq n}\left\vert\frac{d_i}{d_i^*}-1\right\vert = o_P(1) 
\quad
\text{and}
\quad
\max_{1\leq k, l \leq K}\left\vert \frac{\vct{1}_k^\top \mtx{A} \vct{1}_l}{\vct{1}_k^\top \mtx{M} \vct{1}_l} - 1\right\vert = o_P(1).
\]
\end{lemma}

\begin{proof}
Note that
\[
d_i - d_i^* = \sum_{j=1}^n (A_{ij} - M_{ij}),
\]
in which the summands satisfy \eqref{eq:bernstein}. Then, we have
\[
\var(d_i - d_i^*) = \sum_{j=1}^n \nu(M_{ij}) \leq \frac{1}{c_0} \sum_{j=1}^n M_{ij}.
\]
By Lemma \ref{lem:bernstein_inequality}, for any fixed $\epsilon > 0$, we have
\begin{align*}
\label{eq:di_tail_bd}
\sum_{i=1}^n \P \left(|d_i - d_i^*| \geq \epsilon d_i^* \right) 
&\leq \sum_{i=1}^n 2 \exp\left( -\frac{\epsilon^2 {d_i^*}^2}{2 \left(\frac{1}{c_0} (\sum_{j=1}^n M_{ij}) + R \epsilon d_i^* \right)} \right)
\\
& \leq 2n \exp \left(-C \left(\frac{\epsilon^2}{1 + \epsilon}\right) \theta_{\min}^2  n\right) = O(n^{-3}),
\end{align*}
where the last inequality is due to \eqref{eq:sparseness} in Assumption \ref{ass:DCSBM}. Therefore, with probability $1-O(n^{-3})$, we have $\max_{1 \leq i \leq n}\left\vert\frac{d_i}{d_i^*}-1\right\vert \leq \epsilon$. Given $\epsilon$ can be chosen arbitrarily small, we have 
$\max_{1 \leq i \leq n}\left\vert\frac{d_i}{d_i^*}-1\right\vert = o_P(1)$.

With similar arguments, we have $\max_{1\leq k, l \leq K}\left\vert \frac{\vct{1}_k^\top \mtx{A} \vct{1}_l}{\vct{1}_k^\top \mtx{M} \vct{1}_l} - 1\right\vert = o_P(1)$.
\end{proof}

For unweighted networks, \citet{jin2022optimal} proved that SCORE enjoys the Nonsplitting Property (NSP) in both the underfitting and null cases. In other words, the true communities are refinements of the estimated communities with high probability. This property is essential for analyzing sequential testing approaches, as it reduces the number of possible DCSBM fittings. We give a formal definition of this property.

\begin{definition}[Nonsplitting Property \citep{jin2022optimal}]
Let the ground truth of communities in a network be $\mathcal{N}_1, \ldots, \mathcal{N}_K$. Assume $m \leq K$, and $\widehat{\mathcal{N}}^{(m)}_1, \ldots, \widehat{\mathcal{N}}^{(m)}_m$ are the estimated communities by certain clustering method. We say that these estimated communities satisfy the Nonsplitting Property (NSP), if the true communities $\mathcal{N}_1, \ldots, \mathcal{N}_K$ are a refinement of the estimated ones. In other words, for any $k=1,\ldots, K$, there is exactly one $l=1, \ldots, m$, such that $\mathcal{N}_k \cap \widehat{\mathcal{N}}^{(m)}_l \neq \emptyset$.
\end{definition}

Then, we have the following lemma that guarantees the NSP of SCORE under the weighted DCSBM.
\begin{lemma}
\label{thm:nsp}
Under Assumption \ref{ass:DCSBM} with any fixed $c_0>0$, for any fixed $m \leq K$, SCORE satisfies the NSP with probability $1 - O(n^{-3})$.
\end{lemma}

The proof of this lemma is deferred to Appendix \ref{sec:proof_nsp}, which follows the idea of \cite{jin2022optimal}. A crucial component of the proof is the row-wise bounds of eigenvector perturbations \citep{jin2017estimating, abbe2020entrywise}.

\section{PROOFS OF THE MAIN RESULTS}
\label{sec:proof_main}

Throughout this section, we use $C$ to represent a constant that depends only on $c_0$, and its value may change from line to line.

\subsection{Proof of Theorem \ref{thm:null}}
This subsection aims to prove Theorem \ref{thm:null}, so we assume $m=K$. For notational convenience, we omit the superscripts in the estimators indicating the $m$-th step.

Note that
\[
T_{n, m} = \left|\lambda_{m+1} \left(\widehat{\mtx{\Psi}}^{\frac{1}{2}} \mtx{A} \widehat{\mtx{\Psi}}^{\frac{1}{2}}\right)\right|
\]
and
\[
\widehat{\mtx{\Psi}}^{\frac{1}{2}} \mtx{A} \widehat{\mtx{\Psi}}^{\frac{1}{2}} = \widehat{\mtx{\Psi}}^{\frac{1}{2}} \mtx{M} \widehat{\mtx{\Psi}}^{\frac{1}{2}} +
\widehat{\mtx{\Psi}}^{\frac{1}{2}} (\mtx{A} - \mtx{M}) \widehat{\mtx{\Psi}}^{\frac{1}{2}},
\]
where $\rank \left(\widehat{\mtx{\Psi}}^{\frac{1}{2}} \mtx{M} \widehat{\mtx{\Psi}}^{\frac{1}{2}} \right) \leq K=m$. Then we have
\[
T_{n,m} \leq \left\|\widehat{\mtx{\Psi}}^{\frac{1}{2}} \left(\mtx{A}-\mtx{M}\right) \widehat{\mtx{\Psi}}^{\frac{1}{2}} \right\|.
\]
On the other hand,
\begin{align*}
\widehat{\mtx{\Psi}}^{\frac{1}{2}} \left(\mtx{A}-\mtx{M}\right) \widehat{\mtx{\Psi}}^{\frac{1}{2}} 
&= \left(\widehat{\mtx{\Psi}}^{\frac{1}{2}}\mtx{\Psi}^{-\frac{1}{2}}\right) \left(\mtx{\Psi}^{\frac{1}{2}}(\mtx{A}-\mtx{M})\mtx{\Psi}^{\frac{1}{2}}\right)
\left(\mtx{\Psi}^{-\frac{1}{2}}\widehat{\mtx{\Psi}}^{\frac{1}{2}}\right) .
\end{align*}
In order to show that $T_{n,m} \leq 2+o_p(1)$, it suffices to show the following inequalities
\[
\begin{cases}
\left\|\mtx{\Psi}^{\frac{1}{2}}(\mtx{A}-\mtx{M})\mtx{\Psi}^{\frac{1}{2}} \right\| \leq 2+ o_p(1),
\\
\left\| \widehat{\mtx{\Psi}}^{\frac{1}{2}}\mtx{\Psi}^{-\frac{1}{2}} - \mtx{I}_n \right\| \leq o_p(1).
\end{cases} 
\]


The first inequality follows from Lemma \ref{lem:E_norm} and the second inequality can be shown by the following lemma.


\begin{lemma}
\label{lem:theta_B_hat}
Under the null case $m=K$, Assumption \ref{ass:DCSBM} implies
\[
\max_{1 \leq i \leq n}\left\vert\frac{\hat{\theta}_i}{\theta_i}-1\right\vert = o_P(1), 
\quad 
\max_{1\leq k, l \leq K}\left\vert\frac{\widehat{B}_{kl}}{B_{kl}}-1\right\vert = o_P(1),
\]
\[
\max_{1\leq i, j \leq n}\left\vert\frac{\widehat{M}_{ij}}{M_{ij}}-1\right\vert = o_P(1),
\quad
\max_{1\leq k, l \leq n}\left\vert\frac{\widehat{V}_{kl}}{V_{kl}}-1\right\vert  = o_P(1).
\]
Furthermore, we have
\[
\left\| \widehat{\mtx{\Psi}}^{\frac{1}{2}}\mtx{\Psi}^{-\frac{1}{2}} - \mtx{I}_n \right\| = o_P(1).
\]
\end{lemma}

\begin{proof}
In the null case, the NSP given in Theorem \ref{thm:nsp} implies strong consistency for the recovery of communities. Without loss of generality, denote
$\widehat{\N}_k = \N_k$ and thereby $\vct{1}_{k} = \hat{\vct{1}}_k$ for $k=1, \ldots, K$, with probability $1 - O(n^{-3})$. By combining \eqref{eq:B_kl}, \eqref{eq:theta}, \eqref{eq:theta_hat}  and \eqref{eq:B_hat}, we get
\[
\frac{\theta_i}{\hat{\theta}_i} = \frac{d_i^*}{d_i} \frac{\sum_{j \in \mathcal{N}_k} d_j}{\sum_{j \in \mathcal{N}_k} d_j^*} \sqrt{\frac{\vct{1}_k^\top \mtx{M} \vct{1}_k}{\vct{1}_k^\top \mtx{A} \vct{1}_k}}, 
\]
and
\[
\frac{B_{kl}}{\widehat{B}_{kl}} = \frac{\vct{1}_k^\top \mtx{M} \vct{1}_l}{\vct{1}_k^\top \mtx{A} \vct{1}_l} \sqrt{\frac{\vct{1}_k^\top \mtx{A} \vct{1}_k}{\vct{1}_k^\top \mtx{M} \vct{1}_k}} \sqrt{\frac{\vct{1}_l^\top \mtx{A} \vct{1}_l}{\vct{1}_l^\top \mtx{M} \vct{1}_l}}.
\]
Then by Lemma \ref{lem:degree_concentration}, we have
\[
\max_{1 \leq i \leq n}\left\vert\frac{\hat{\theta}_i}{\theta_i}-1\right\vert = o_P(1), 
\quad 
\max_{1\leq k, l \leq K}\left\vert\frac{\widehat{B}_{kl}}{B_{kl}}-1\right\vert = o_P(1).
\]
Furthermore, $\max_{1\leq i, j \leq n}\left\vert\frac{\widehat{M}_{ij}}{M_{ij}}-1\right\vert = o_P(1)$ follows from 
\[
\frac{\widehat{M}_{ij}}{M_{ij}} = \frac{\hat{\theta}_i \hat{\theta}_j \widehat{B}_{\phi(i)\phi(j)}}{\theta_i \theta_j B_{\phi(i)\phi(j)}}.
\]
Then 
\begin{equation}
\label{eq:variance_concentration}
\max_{1\leq i, j \leq n}\left\vert\frac{\widehat{V}_{ij}}{V_{ij}}-1\right\vert = o_P(1)
\end{equation}
follows from the assumption \eqref{eq:variance_linearity} on the variance-mean function. 

On the other hand, notice that \eqref{eq:variance_concentration} implies
\[
\max_{1\leq i \leq n} \left\vert\frac{ \sum_{j=1}^n \widehat{V}_{ij} \psi_i \psi_j} {\sum_{j=1}^n V_{ij} \psi_i \psi_j }-1\right\vert = o_P(1).
\]
By \eqref{eq:pop_scaling}, we further have
\[
\max_{1\leq i \leq n} \left\vert \sum_{j=1}^n \widehat{V}_{ij} \psi_i \psi_j -1\right\vert = o_P(1).
\]
Denote $V_{\max} = \max_{1 \leq i, j \leq n} V_{ij}$ and $V_{\min} = \min_{1 \leq i, j \leq n} V_{ij}$. Additionally, $\widehat{V}_{\max}$ and $\widehat{V}_{\min}$ are defined similarly. Assumption \ref{ass:DCSBM} implies 
\[
\left. V_{\max}\middle/V_{\min} \right. = O(1) \quad \text{and} \quad \left. V_{\min}\middle/V_{\max} \right. = O(1). 
\]
Combined with \eqref{eq:variance_concentration}, we have
\[
\left. \widehat{V}_{\max}\middle/V_{\max} \right. = O_P(1) \quad \text{and} \quad \left. \widehat{V}_{\min}\middle/V_{\min} \right. = O_P(1).
\]
Plug in the above inequalities to Lemma \ref{lem:landa_9}, we have $\left\| \widehat{\mtx{\Psi}}^{\frac{1}{2}}\mtx{\Psi}^{-\frac{1}{2}} - \mtx{I}_n \right\| = o_P(1)$.
\end{proof}

\subsection{Proof of Theorem \ref{thm:under-fitting}}

This subsection aims to prove Theorem \ref{thm:under-fitting} with $m<K$. Again, we omit the superscripts in the estimators.

Note that
\begin{equation}
\label{eq:underfitting_lower}
T_{n, m} = \left|\lambda_{m+1} \left(\widehat{\mtx{\Psi}}^{\frac{1}{2}} \mtx{A} \widehat{\mtx{\Psi}}^{\frac{1}{2}}\right)\right| \geq \left|\lambda_{K} \left(\widehat{\mtx{\Psi}}^{\frac{1}{2}} \mtx{A} \widehat{\mtx{\Psi}}^{\frac{1}{2}}\right)\right| \geq |\lambda_K(\mtx{A})| \hat{\psi}_{\min} \geq (|\lambda_K(\mtx{M})| - \|\mtx{A} - \mtx{M}\|) \hat{\psi}_{\min},
\end{equation}
where $\hat{\psi}_{\min} = \min_{1 \leq i \leq n} \hat{\psi}_i$. Therefore, it suffices to find lower bounds for $|\lambda_K(\mtx{M})|$ and $\hat{\psi}_{\min}$, as well as an upper bound for $\|\mtx{A} - \mtx{M}\|$.

\begin{lemma}
\label{lem:theta_hat_underfitting}
For the underfitting case $m<K$, denote $\widehat{M}_{\max} = \max_{1 \leq i, j \leq n} \widehat{M}_{ij}$ and $\widehat{M}_{\min} = \min_{1 \leq i, j \leq n} \widehat{M}_{ij}$. Also, denote $\widehat{V}_{\max}$ and $\widehat{V}_{\min}$ in a similar manner. Then under Assumption \ref{ass:DCSBM}, there holds
\[
\max \left(\frac{\widehat{M}_{\max}}{\theta_{\min}^2},~~\frac{\theta_{\min}^2}{\widehat{M}_{\min}} \right) = O_P(1),
\]
and
\[
\max \left(\frac{\widehat{V}_{\max}}{\theta_{\min}^2},~~\frac{\theta_{\min}^2}{\widehat{V}_{\min}} \right) = O_P(1).
\]
Furthermore, we have
\[
\frac{1}{\theta_{\min} \hat{\psi}_{\min} \sqrt{n}} = O_P(1).
\]
\end{lemma}



\begin{proof}
In the underfitting case $m<K$, the NSP given in Theorem \ref{thm:nsp} implies that the true communities $\N_1, \ldots, \N_K$ are refinements of the estimated communities $\widehat{\N}_1, \ldots, \widehat{\N}_m$. For each $1\leq k \leq m$, we assume that the number of true communities contained in $\widehat{\N}_k$ is $r_k\geq 1$, which implies that $r_1 + \cdots + r_m = K$. Then, we can represent the estimated communities as
\[
\widehat{\mathcal{N}}_k = \mathcal{N}_{h_{k1}} \cup \cdots\cup \mathcal{N}_{h_{kr_k}}, \quad k=1, \ldots, m.
\]
Here all indices $h_{kj}$ for $k=1, \ldots, m$ and $j=1, \ldots, r_k$ are distinct over $1, \ldots, K$. Similarly, we can decompose
\[
\hat{\vct{1}}_k = \vct{1}_{h_{k1}} + \cdots + \vct{1}_{h_{k r_k}}, \quad k=1, \ldots, m.
\]

Recall that
\[
\widehat{M}_{ij} = \frac{(\hat{\vct{1}}_k)^\top \mtx{A} \hat{\vct{1}}_l}{(\hat{\vct{1}}_k)^\top \mtx{A} \vct{1}_n  \cdot  (\hat{\vct{1}}_l)^\top \mtx{A} \vct{1}_n } d_i d_j.
\]


Then by Lemma \ref{lem:degree_concentration}, NSP and Assumption \ref{ass:DCSBM}, it is easy to obtain 
\[
\max \left(\frac{\widehat{M}_{\max}}{\theta_{\min}^2},~\frac{\theta_{\min}^2}{\widehat{M}_{\min}} \right) = O_P(1).
\]
Furthermore, by \eqref{eq:variance_linearity} in Assumption \ref{ass:DCSBM}, we have
\[
\max \left(\frac{\widehat{V}_{\max}}{\theta_{\min}^2},~\frac{\theta_{\min}^2}{\widehat{V}_{\min}} \right) = O_P(1).
\]
Then by \eqref{eq:sample_scaling} and Lemma \ref{lem:landa_2}, we have 
\[
\frac{1}{\theta_{\min} \hat{\psi}_{\min} \sqrt{n}} = O_P(1).
\]
\end{proof}

\begin{proof}[Proof of Theorem \ref{thm:under-fitting}]
By Assumption \ref{ass:DCSBM}, we have
\[
\theta_{\min}\sqrt{n} \geq \log^3 n \rightarrow \infty.
\]
Combined with Lemma \ref{lem:spectral_lowerbound} and \ref{lem:noise_norm}, Theorem \ref{thm:under-fitting} is proved from \eqref{eq:underfitting_lower}.

\end{proof}

\section{PROOF OF THE NONSPLITTING PROPERTY}
\label{sec:proof_nsp}
The proof of the NSP of SCORE, i.e., Lemma \ref{thm:nsp}, basically follows the arguments in \citet{jin2017estimating} and \citet{jin2022optimal}. In other words, Lemma \ref{thm:nsp} is an extension of Theorem 3.2 of \citet{jin2022optimal} to the case of weighted DCSBM. To carry out this extension, we first need some probabilistic results for subexponential random vectors and matrices.

\subsection{Preliminaries}
\begin{lemma}[Vector Bernstein-type Inequality, Theorem 2.5 of \citet{Bosq2000}]
\label{lem:vector_bernstein}
If $\{\vct{X}_i\}_{i=1}^n$ are independent random vectors in a separable Hilbert space (where the norm is denoted by $\|\cdot\|$) with $\E[\vct{X}_i] = \vct{0}$ and 
\[
\sum_{i=1}^n \E \|\vct{X}_i\|^p \leq \frac{p!}{2} \sigma^2 R^{p-2}, \quad p=2, 3, 4, \ldots .
\]
Then, for all $t>0$,
\[
\P \left(\left\| \sum_{i=1}^n \vct{X}_i \right\| \geq t\right) \leq 2 \exp \left(- \frac{t^2}{2(\sigma^2 + R t)}\right) .
\]
\end{lemma}

The following lemma is an application of the vector Bernstein-type inequality supporting the proof of Lemma \ref{lem:row_wise_bound_matrix}.

\begin{lemma}
\label{lem:vector_bernstein_application}
Let $X_1, \ldots, X_n$ be independent random variables satisfying 
\[
\E[|X_i - \mu_i|^p] \leq C_0\left(\frac{p!}{2}\right) R^{p-2} \mu_i,
\]
where $\mu_{\max} \coloneqq \max_{1 \leq i \leq n} \mu_i \leq C_1$. Let $\vct{w}_1, \ldots, \vct{w}_n \in \mathbb{R}^d$ be fixed vectors. Then, with probability $1 - O(n^{-4})$, we have
\begin{equation}
\label{eq:vector_bernstein_application}
\left\| \sum_{i=1}^n (X_i - \mu_i) \vct{w}_i \right\|_2 \leq C  \left(\sqrt{\mu_{\max}}\|\mtx{W}\|_F \sqrt{\log n} + \|\mtx{W}\|_{2 \rightarrow \infty} (\log n)\right),
\end{equation}
where $\mtx{W}^\top = [\vct{w}_1, \ldots, \vct{w}_n]$, $\|\mtx{W}\|_{2 \rightarrow \infty} = \max_{1 \leq i \leq n} \|\vct{w}_i\|_2$, and $C$ is a constant that only relies on $C_0$, $R$, and $C_1$.
\end{lemma}

\begin{proof}
For any $p\geq 2$, we have
\begin{align*}
\sum_{i=1}^n \E \|(X_i - \mu_i) \vct{w}_i \|_2^p
&= \sum_{i=1}^n (\E|X_i -\mu_i|^p) \|\vct{w}_i\|_2^p
\\
& \leq \sum_{i=1}^n C_0 \left(\frac{p!}{2}\right) R^{p-2} \mu_i \|\vct{w}_i\|_2^p
\\
& \leq \sum_{i=1}^n C_0 \left(\frac{p!}{2}\right) R^{p-2} C_1 \|\vct{w}_i\|_2^2 \|\mtx{W}\|_{2 \rightarrow \infty}^{p-2}
\\
& = \left(\frac{p!}{2}\right) \left(R\|\mtx{W}\|_{2 \rightarrow \infty}\right)^{p-2} (C_0 C_1 \|\mtx{W}\|_F^2),
\end{align*}
Then, by Lemma \ref{lem:vector_bernstein}, for any $t>0$,
\[
\P\left(\left\| \sum_{i=1}^n (X_i -\lambda_i)\vct{w}_i\right\|_2 \geq t\right)
\leq 2 \exp \left(-\frac{t^2}{C \mu_{\max} \|\mtx{W}\|_F^2 + (R \|\mtx{W}\|_{2 \rightarrow \infty}) t} \right).
\]
Then we get \eqref{eq:vector_bernstein_application} with probability $1-O(n^{-4})$ for sufficiently large $C$.
\end{proof}

The following lemma is the subexponential case of the matrix Bernstein inequality.

\begin{lemma}[Theorem 6.2 of \citet{tropp2012user}]
\label{thm:matrix_Bernstein_subexponential}
Consider a finite sequence $\{\mtx{X}_k\}$ of independent, random, symmetric matrices with dimension $d$. Assume that 
\[
\E[\mtx{X}_k] = \mtx{0} ~~\text{and}~~ \E[\mtx{X}_k^p]   \preceq \frac{p!}{2} R^{p-2} \mtx{A}_k^2 , \quad  p=2, 3, 4, \ldots.
\]
Compute the variance parameter
\[
\sigma^2 \coloneqq \left\|\sum_k \mtx{A}_k^2\right\|.
\]
Then, the following chain of inequalities holds for all $t \geq 0$:
\begin{align*}
\P\left(\lambda_{\max}\left(\sum_k \mtx{X}_k\right) \geq t\right)
&\leq d \exp\left(-\frac{t^2}{2(\sigma^2 + Rt)}\right). 
\end{align*}
\end{lemma}

The following theorem from \citet{abbe2020entrywise} provides perturbation bounds of eigenspaces which is crucial to the proof of Theorem \ref{thm:nsp}. 

\begin{lemma}[\citet{abbe2020entrywise}]
\label{thm:row-wise-Abbe}
Suppose $\mtx{A} \in \mathbb{R}^{n \times n}$ is a symmetric random matrix, and let $\mtx{A}^* = \E[\mtx{A}]$. Denote the eigenvalues of $\mtx{A}$ by $\lambda_1 \geq \cdots \geq \lambda_n$, and their associated eigenvectors by $\{\vct{u}_j\}_{j=1}^n$. Analogously for $\mtx{A}^*$, the eigenvalues and eigenvectors are denoted by $\lambda_1^* \geq \cdots \geq \lambda_n^*$ and $\{\vct{u}_j^*\}_{j=1}^n$. For convenience, we also define $\lambda_0 = \lambda_0^* = \infty$ and $\lambda_{n+1} = \lambda_{n+1}^* = -\infty$. Note that we allow eigenvalues to be identical, so some eigenvectors may be defined up to rotations.

Suppose $r$ and $s$ are two integers satisfying $1 \leq r \leq n$ and $0 \leq s \leq n-r$. Let $\mtx{U} = [\vct{u}_{s+1}, \ldots, \vct{u}_{s+r}] \in \mathbb{R}^{n\times r}$, $\mtx{U}^* = [\vct{u}_{s+1}^*, \ldots, \vct{u}_{s+r}^*] \in \mathbb{R}^{n \times r}$ and $\mtx{\Lambda}^* = \diag(\lambda_{s+1}^*, \ldots, \lambda_{s+r}^*) \in \mathbb{R}^{r \times r}$. Define
\[
\Delta^* = (\lambda_s^* - \lambda_{s+1}^*) \wedge (\lambda_{s+r}^* - \lambda_{s+r+1}^*) \wedge \min_{1 \leq i \leq r} |\lambda_{s+i}^*|
\]
and
\[
\kappa = \max_{1\leq i \leq r} \left.|\lambda_{s+i}^*|\middle/\Delta^* \right..
\]

Suppose for some $\gamma \geq 0$, the following assumptions hold:
\begin{enumerate}
    \item[A1] (Incoherence) $\|\mtx{A}^*\|_{2 \rightarrow \infty} \leq \gamma \Delta^*$.
    \item[A2] (Row- and column-wise independence) For any $m \in [n]$, the entries in the $m$-th row and column of $\mtx{A}$ are independent with other entries: namely, $\{A_{ij}, i=m \text{~or~} j=m\}$ are independent of $\{A_{ij}: i \neq m, j \neq m\}$.
    \item[A3] (Spectral norm concentration) $32 \kappa \max\{\gamma, \varphi(\gamma)\} \leq 1$ and for some $\delta_0 \in (0, 1)$,
    \[
    \P\left(\|\mtx{A}-\mtx{A}^*\| \leq \gamma \Delta^* \right) \geq 1 - \delta_0.
    \]
    \item[A4] (Row concentration) Suppose $\varphi(x)$ is continuous and non-decreasing in $\mathbb{R}_+$ with $\varphi(0)=0$, $\varphi(x)/x$ is non-increasing in $\mathbb{R}_+$, and $\delta_1 \in (0, 1)$. For any $i \in [n]$ and $\mtx{W} \in \mathbb{R}^{n \times r}$,
    \[
    \P\left(\|(\mtx{A} - \mtx{A}^*)_{i \cdot} \mtx{W}\|_2 \leq \Delta^* \|\mtx{W}\|_{2 \rightarrow \infty} \varphi\left(\frac{\|\mtx{W}\|_F}{\sqrt{n} \|\mtx{W}\|_{2 \rightarrow \infty}}\right)\right) \geq 1 - \frac{\delta_1}{n}.
    \]
\end{enumerate}

Under Assumptions (A1)---(A4), with probability at least $1-\delta_0 - 2\delta_1$, there exists an orthogonal matrix $\mtx{Q}$ such that
\begin{equation}
\label{eq:rowwise_first_order}
\|\mtx{U}\mtx{Q} - \mtx{A}\mtx{U}^*(\mtx{\Lambda}^*)^{-1})\|_{2 \rightarrow \infty} \lesssim \kappa (\kappa + \varphi(1))(\gamma + \varphi(\gamma))\|\mtx{U}^*\|_{2 \rightarrow \infty} + \gamma\|\mtx{A}^*\|_{2 \rightarrow \infty}\left.\middle/\Delta^* \right..
\end{equation}
Here $\lesssim$ only hides absolute constants.
\end{lemma}

\subsection{Spectral Properties for the Adjacency Matrix}
Now, we introduce some spectral properties of the mean adjacency matrix $\mtx{M}$ and the observed weighted adjacency matrix $\mtx{A}$, which are useful in proving the NSP of SCORE. We begin by studying the eigenvalues and eigenvectors of $\mtx{M}$. 
The population adjacency matrix $\mtx{M}$ is a rank-$K$ matrix, with nonzero eigenvalues $\lambda_1^*, \lambda_2^*, \ldots , \lambda_K^*$ sorted in descending magnitude and their corresponding unit-norm eigenvectors $\vct{u}_1^*, \ldots, \vct{u}_K^*$. Notice that by Perron's theorem \citep{horn2012matrix}, $\lambda_1^*$ is positive with multiplicity 1, and we can choose $\vct{u}_1^*$ such that all its entries are strictly positive. Denote $\mtx{U}^* = [\vct{u}_2^*, \ldots, \vct{u}_K^*] \in \mathbb{R}^{n \times (K-1)}$ and $\left(\mtx{U}^*_{i\cdot} \right)^\top$ as its $i$-th row. Denote $\mtx{\Lambda}^* = \diag ( \lambda_2^*, \ldots , \lambda_K^* ) \in \mathbb{R}^{(K-1) \times (K-1)}$.

\begin{lemma}[Lemma B.1 of \citet{jin2017estimating}]
\label{lem:population_eigen}
Under Assumption \ref{ass:DCSBM}, the following statements are true:
\begin{enumerate}
    \item $|\lambda_k^* | \asymp \|\vct{\theta}\|_2^2$, $ 1\leq k\leq K$.
    \item $\lambda_1^* - |\lambda_2^*| \asymp \lambda_1^* $.
\end{enumerate}
\end{lemma}

\begin{proof}
This result is basically the same as Lemma B.1 of \citet{jin2017estimating} under Assumption \ref{ass:DCSBM}. In particular, $\lambda_1^* - |\lambda_2^*| \asymp \lambda_1^* $ can be straightforwardly obtained from \citet{https://doi.org/10.48550/arxiv.1906.04875}.
\end{proof}

Also, we have the following properties of the eigenvectors of $\mtx{M}$:

\begin{lemma}[Lemma B.2 of \citet{jin2017estimating}]
\label{lem:population_eigen_vec}
Under Assumption \ref{ass:DCSBM}, the following statements are true:
\begin{enumerate}
    \item If we choose the sign of $\vct{u}_1^*$ such that $\sum_{i=1}^n u_{1}^*(i) > 0$, then the entries of $\vct{u}_1^*$ are positive satisfying $C^{-1} \theta_i/\|\vct{\theta}\|_2 \leq u_{1}^*(i) \leq C\theta_i/\|\vct{\theta}\|_2$, $1\leq i\leq n$.
    \item $\| \mtx{U}^*_{i\cdot} \|_2 \leq C\sqrt{K}\theta_i/\|\vct{\theta}\|_2 $, $ 1\leq i \leq n$.
\end{enumerate}
Here, $C$ is a constant that depends only on $c_0$.
\end{lemma}

Now, we come back to community detection with SCORE. Let $|\lambda_1| \geq |\lambda_2| \geq \cdots \geq |\lambda_K|$ be the leading $K$ eigenvalues of $\mtx{A}$ in magnitude, with corresponding unit-norm eigenvectors $\vct{u}_1, \ldots, \vct{u}_K$. Denote $\mtx{U} = [\vct{u}_2, \ldots, \vct{u}_K] \in \mathbb{R}^{n \times (K-1)}$. The proof of the following lemma is similar to that of Lemma 2.1 in \citet{jin2017estimating}, but requires some adaptions for the weighted DCSBM. We defer the proof to Section \ref{sec:proof_row_wise}.

\begin{lemma}
\label{lem:row_wise_bound_matrix}
Under Assumption \ref{ass:DCSBM}, with probability $1 - O(n^{-3})$, the following statements are true:
\begin{enumerate}
    \item We can select $\vct{u}_1$ such that $\|\vct{u}_1 - \vct{u}_1^*\|_\infty \leq \frac{C(c_0)}{\sqrt{n \log n}}$.
    \item $\|\mtx{U}\mtx{Q}^\top - \mtx{U}^*\|_{2 \rightarrow \infty}\leq \frac{C(c_0)}{\sqrt{n \log n}}$ for some $\mtx{Q} \in \mathcal{O}_{K-1}$.
\end{enumerate}
\end{lemma}

\subsection{Proof of Lemma \ref{thm:nsp}}
Using the above lemmas, we can prove Lemma \ref{thm:nsp} by following the arguments in \citet{jin2022optimal}. Here, we give a brief outline.

Define $\mtx{R}^{(K)}$ as an $n \times (K-1)$ matrix constructed from the eigenvectors $\vct{u}_1$, $\vct{u}_2$, ..., $\vct{u}_K$ by taking entrywise ratios of $\vct{u}_2$,..., $\vct{u}_K$ to $\vct{u}_1$, i.e., $R^{(K)}(i, k) = u_{k+1}(i)/u_1(i)$ for $1\leq i \leq n$ and $0\leq k \leq K-1$. For any $2 \leq m \leq K$, let $\mtx{R}^{(m)}$ be an $n \times (m-1)$ matrix consists of the first $m-1$ columns of $\mtx{R}^{(K)}$. For any candidate number of clusters $m$,  SCORE amounts to performing $k$-means clustering on the rows of $\mtx{R}^{(m)}$. Therefore, we need to study the statistical properties of the rows of $\mtx{R}^{(m)}$.

Next, we define the population counterparts. Denote $\mathcal{O}_{K-1}$ as the space of $(K-1) \times (K-1)$ orthogonal matrices.
For any $\mtx{Q} \in \mathcal{O}_{K-1}$ and any $2 \leq k \leq K$, let $\vct{u}_k^*(\mtx{Q})$ be the $(k-1)$-th column of $[\vct{u}_2^*, \ldots, \vct{u}_K^*]\mtx{Q}$. Define $\mtx{R}^{*(K)}(\mtx{Q}) \in \mathbb{R}^{n \times (K-1)}$ such that its $(k-1)$-th column is the entrywise ratio of $\vct{u}_k^*(\mtx{Q})$ to $\vct{u}_1^*$. For any $2 \leq m \leq K$, let $\mtx{R}^{*(m)}(\mtx{Q}) \in \mathbb{R}^{n \times (m-1)}$ consist of the first $m-1$ columns of $\mtx{R}^{*(K)}(\mtx{Q})$. The following lemma provides a uniform upper bound on the difference between the corresponding rows of $\mtx{R}^{(m)}$ and $\mtx{R}^{*(m)}(\mtx{Q})$.

\begin{lemma}[Row-wise deviation bound]
\label{lem:row_wise_bound}
For $2\leq m \leq K$, denote $\left( \vct{r}_i^{(m)} \right)^\top$ and $\left( \vct{r}_i^{*(m)}(\mtx{Q}) \right)^\top$ as the $i$-th row of $\mtx{R}^{(m)}$ and $\mtx{R}^{*(m)}(\mtx{Q})$, respectively.
Under Assumption \ref{ass:DCSBM}, with probability $1 - O(n^{-3})$, there exists a $(K-1)\times (K-1)$ orthogonal matrix $\mtx{Q}$, such that
\begin{equation}
\label{eq:row_wise_bound}
\|\vct{r}_i^{(m)} - \vct{r}_i^{*(m)}(\mtx{Q})\|_{2} 
\leq \|\vct{r}_i^{(K)} - \vct{r}_i^{*(K)}(\mtx{Q})\|_{2} \leq \frac{C(c_0)}{\sqrt{\log n}},
\end{equation}
for each $2 \leq m \leq K$ and $1 \leq i \leq n$.
\end{lemma}

This lemma is a straightforward corollary of Lemmas \ref{lem:population_eigen_vec} and \ref{lem:row_wise_bound_matrix} by following the proof of Lemma 4.1 in \citet{jin2022optimal}, so we omit the details.

Combining Lemma \ref{lem:row_wise_bound} with Lemma 4.2, Lemma 4.3 and Theorem 4.1 in \citet{jin2022optimal}, Lemma \ref{thm:nsp} is proved by exactly the same arguments as in \citet{jin2022optimal}.

\subsection{Proof of Lemma \ref{lem:row_wise_bound_matrix}}
\label{sec:proof_row_wise}
\begin{proof}
Divide $\lambda_1^*,\ldots,\lambda_K^*$ into three groups: (i) $\lambda_1^*$, (ii) positive values in $\lambda_2^*,\ldots,\lambda_K^*$, and (\text{iii}) negative values in $\lambda_2^*,\ldots,\lambda_K^*$.
We shall apply Lemma \ref{thm:row-wise-Abbe} to all three groups. For succinctness, we only show in detail the application to group (ii), while the proofs for the other two groups are similar and thus omitted. 

Denote $K_1$ as the number of eigenvalues in group (ii). Define $\mtx{\Lambda}_1^*$ as the diagonal matrix consisting of eigenvalues in group (ii), and $\mtx{U}_1^*$
as the matrix whose columns are the associated eigenvectors. Define the empirical counterparts of the two matrices as $\mtx{\Lambda}_1$ and $\mtx{U}_1$. To show the second statement in the lemma, we aim to first show with high probability that
\begin{enumerate}[label=(\alph*)]
\item $\|\mtx{U}_1\mtx{Q}^\top - \mtx{A}\mtx{U}_1^*(\mtx{\Lambda}_1^*)^{-1}\|_{2 \rightarrow \infty}\leq \frac{C(c_0)}{\sqrt{n \log n}}$ for some $\mtx{Q} \in \mathcal{O}_{K_1}$.
\item $\|\mtx{U}_1^* - \mtx{A}\mtx{U}_1^*(\mtx{\Lambda}_1^*)^{-1}\|_{2 \rightarrow \infty} \leq \frac{C(c_0)}{\sqrt{n \log n}}$.
\end{enumerate}

\paragraph{Proof of (a)}

To apply Lemma \ref{thm:row-wise-Abbe}, we need to determine $\gamma$ and $\varphi: \mathbb{R}_+ \rightarrow \mathbb{R}_+$, and then verify Assumption (A1)---(A4) as required by the theorem. Note that by Lemma \ref{lem:population_eigen}, we have 
\[
\Delta^* = \min\left\{\lambda_1^* - \lambda_2^*, |\lambda_K^*| \right\} \geq C(c_0) \theta_{\min}^2   n,
\]
and $\kappa \leq C(c_0)$.
We choose an appropriately large $C_1(c_0)$ and let
\[
\gamma = \frac{C_1(c_0)}{\sqrt{\log n}} .
\]
\begin{enumerate}
    \item[A1] This is trivial since 
    \[
    \|\mtx{M}\|_{2\rightarrow \infty} \leq c_0^{-2} \theta_{\min}^2  \sqrt{n} \leq \Delta^* \gamma,
    \]
    when $n$ is sufficiently large.
    \item[A2] This obviously holds.
    \item[A3] In the proof of A4, we show that $\varphi(\gamma) \leq C(c_0) (\log n)^{-3/2}$, hence, $32 \kappa \max\{\gamma, \varphi(\gamma)\} \leq 1$ when $n$ is sufficiently large. 
    
   Next, we show the spectral norm perturbation bound by applying the subexponential case of matrix Bernstein inequality (Theorem 6.2 of \citet{tropp2012user}, restated as Lemma \ref{thm:matrix_Bernstein_subexponential} in this paper).
Let $\mtx{X}^{ij} = (A_{ij}-M_{ij})(\mtx{E}^{ij} + \mtx{E}^{ji})$ for $i<j$, and $\mtx{X}^{ii} = (A_{ii}-M_{ii})\mtx{E}^{ii}$ for $i=1,\ldots,n$, where $\mtx{E}^{ij}$ is a $n\times n$ matrix with 1 on the $(i,j)$-th entry and 0 elsewhere. Obviously, $\E [\mtx{X}^{ij}] = \mtx{0}$. Notice that 
    \begin{align*}
    \E[(\mtx{X}^{ij})^p]=
        \begin{cases}
        \E[(A_{ij}-M_{ij})^p](\mtx{E}^{ij} + \mtx{E}^{ji})\quad \text{when $p$ is odd};
        \\
        \E[(A_{ij}-M_{ij})^p](\mtx{E}^{ij} + \mtx{E}^{ji})\quad \text{when $p$ is even}.
        \end{cases}
    \end{align*}
Also, note that $-(\mtx{E}^{ii} + \mtx{E}^{jj}) \preceq  \mtx{E}^{ij} + \mtx{E}^{ji} \preceq \mtx{E}^{ii} + \mtx{E}^{jj}$. Then, by \eqref{eq:bernstein} in Assumption \ref{ass:DCSBM}, for any integer $p \geq 2$, we have
\[
\left| \E[(A_{ij}-M_{ij})^p] \right| \leq \E[|A_{ij}-M_{ij}|^p] \leq C' \left(\frac{p!}{2}\right) R^{p-2} M_{ij} ,
\]
where $C'$ and $R$ only depend on $c_0$. Then,
    \begin{align*}
    \E[(\mtx{X}^{ij})^p] \preceq C' \frac{p!}{2}R^{p-2} M_{ij} (\mtx{E}^{ii} + \mtx{E}^{jj}), \quad p=2, 3, 4, \ldots.
    \end{align*}
    Thus, the conditions of Lemma \ref{thm:matrix_Bernstein_subexponential} are verified. Notice that $\sum_{i\leq j} \mtx{X}^{ij} = \mtx{A} - \mtx{M}$. Denote 
    \begin{align*}
        \sigma^2 = \left\| \sum_{i\leq j} C'  M_{ij} (\mtx{E}^{ii} + \mtx{E}^{jj}) \right\| = C' \left(\underset{i}{\max} \sum_{j=1}^n M_{ij}\right) \lesssim \theta_{\min}^2  n,
    \end{align*}
   where $\lesssim$ only hides a constant depending on $c_0$. Then, for all $t\geq 0$,
    \begin{align*}
    \P\left(\| \mtx{A} - \mtx{M} \| \geq t\right)
    \leq n \exp\left(-\frac{
    t^2}{2(\sigma^2 + Rt )}\right)
    \leq n \exp\left(-\frac{1}{4} \left(\frac{t^2}{\sigma^2}\wedge \frac{t}{R}\right)\right).
    \end{align*}
    By (\ref{eq:sparseness}) in Assumption \ref{ass:DCSBM}, when $n$ is sufficiently large, with probability $1-O(n^{-3})$, for a sufficiently large $C(c_0)$,
    \begin{align*}
        \| \mtx{A} - \mtx{M} \| \leq C(c_0) \theta_{\min} \sqrt{n\log n}.
    \end{align*}
 
    \item[A4] 
    By Lemma \ref{lem:vector_bernstein_application}, for any $1\leq i\leq n$ and $\mtx{W} \in \mathbb{R}^{n \times r}$, with probability $1 - O(n^{-4})$, there holds
    \begin{equation}
    \label{eq:A4_original}
    \| (\mtx{A} - \mtx{M})_{i \cdot} \mtx{W} \|_2 \leq C_2(c_0) \max\left(\theta_{\min}\|\mtx{W}\|_F\sqrt{\log n}, \|\mtx{W}\|_{2 \rightarrow \infty} (\log n)  \right),
    \end{equation}
    for a sufficiently large $C_2(c_0)$.
    Since $\Delta^* \geq C(c_0) \theta_{\min}^2   n$, we have
    \begin{align*}
        \|(\mtx{A} - \mtx{M})_{i \cdot} \mtx{W}\|_2 
        &\leq C_2(c_0) \Delta^* \|\mtx{W}\|_{2 \rightarrow \infty} \max \left( \frac{\sqrt{n \log n}}{\theta_{\min}  n}\frac{\|\mtx{W}\|_F}{\sqrt{n}\|\mtx{W}\|_{2 \rightarrow \infty}}, \frac{\log n}{\theta_{\min}^2  n} \right).
    \end{align*}
    Now we follow similar arguments as in the proof of Lemma 2.1 of \citet{jin2017estimating}. 
    Define the quantities $t_1=C_2(c_0)\left(\theta_{\min}  n \right)^{-1} \sqrt{n \log n}$ and $t_2=C_2(c_0)\left(\theta_{\min}^2 n  \right)^{-1} \log n$ .
    Define the function
    \begin{align*}
        \widetilde{\varphi}(x) = \max\left(t_1 x, t_2\right).
    \end{align*}
    Then, we have for any $1\leq i\leq n$, with probability $1 - O(n^{-4})$,
    \begin{equation}
    \label{eq:A4}
        \|(\mtx{A} - \mtx{M})_{i \cdot} \mtx{W}\|_2 
        \leq \Delta^* \|\mtx{W}\|_{2 \rightarrow \infty} 
        \widetilde{\varphi}\left( \frac{\|\mtx{W}\|_F}{\sqrt{n}\|\mtx{W}\|_{2 \rightarrow \infty}} \right).
    \end{equation}
    First, notice that $(\sqrt{n}\|\mtx{W}\|_{2 \rightarrow \infty})^{-1} \|\mtx{W}\|_F \in \left[n^{-1/2}, 1\right]$. 
    Next, observe that by (\ref{eq:sparseness}), when $n$ is sufficiently large, we have $t_2 / t_1 = (\theta_{\min}\sqrt{n})^{-1}\sqrt{\log n} > n^{-1/2}$.
    Therefore, when $x\in \left[n^{-1/2}, t_2/t_1 \right]$, $t_1 x \leq t_2$, i.e., $\widetilde{\varphi}(x) = t_2$; when $x\in \left(t_2/t_1, 1 \right]$, $t_1 x > t_2$, i.e., $\widetilde{\varphi}(x) = t_1 x$. As a result, we can construct $\varphi(\cdot)$ as: 
    \begin{align*}
    \varphi(x)=
        \begin{cases}
        \sqrt{n} t_2 x \quad &\text{for}~ 0\leq x\leq n^{-1/2};
        \\
         t_2 \quad &\text{for}~  n^{-1/2} < x\leq t_2 / t_1;
        \\
        t_1 x \quad &\text{for}~ t_2/t_1 < x\leq 1;
        \\
        t_1 \quad &\text{for}~ x > 1.
        \end{cases}
    \end{align*}
    Obviously, $\varphi(x)$ is continuous and non-decreasing in $\mathbb{R}_+$, with $\varphi(0)=0$ and $\varphi(x)/x$ being non-increasing in $\mathbb{R}_+$. By (\ref{eq:A4}) and $0\leq \widetilde{\varphi}(x) \leq \varphi (x)$, we have with probability $1 - O(n^{-4})$,
    \begin{equation*}
        \|(\mtx{A} - \mtx{M})_{i \cdot} \mtx{W}\|_2 
        \leq \Delta^* \|\mtx{W}\|_{2 \rightarrow \infty} 
        \varphi\left( \frac{\|\mtx{W}\|_F}{\sqrt{n}\|\mtx{W}\|_{2 \rightarrow \infty}} \right).
    \end{equation*}
    By the definition of $t_1$ and (\ref{eq:sparseness}), we have
    $t_1 \leq (\log n)^{-5/2}$.
    Furthermore, since $\varphi(x)\leq t_1$, we have $\varphi(\gamma) \leq (\log n)^{-5/2}$.
\end{enumerate}

After verifying (A1)---(A4), we obtain the bound \eqref{eq:rowwise_first_order}. 
Note that based on the definition of $\varphi(x)$, we have $\kappa(\kappa + \varphi(1)) \leq C(c_0)$ and $\gamma + \varphi(\gamma) \leq \frac{C(c_0)}{\sqrt{\log n}}$.
Also, by Lemma \ref{lem:population_eigen_vec}, $\|\mtx{U}^*_1\|_{2 \rightarrow \infty} \leq  \frac{C(c_0)}{\sqrt{n}}$. Since $\|\mtx{M}\|_{2\rightarrow \infty} \leq c_0^{-2} \theta_{\min}^2 \sqrt{n}$ and $\Delta^* \geq C(c_0)  \theta_{\min}^2  n$, we have $\|\mtx{M} \|_{2 \rightarrow \infty}/\Delta^* \leq \frac{C(c_0)}{\sqrt{n}}$.
Finally, we obtain that with probability $1-O(n^{-3})$,
\begin{align*}
    \left\|\mtx{U}_1\mtx{Q} - \mtx{A}\mtx{U}^*_1(\mtx{\Lambda}^*_1)^{-1}\right\|_{2 \rightarrow \infty} &\lesssim \kappa (\kappa + \varphi(1))(\gamma + \varphi(\gamma))\left\|\mtx{U}^*_1\right\|_{2 \rightarrow \infty} + \gamma \left. \|\mtx{M}\|_{2 \rightarrow \infty}\middle/\Delta^* \right. \\
    & \lesssim \frac{1}{\sqrt{ n\log n}},
\end{align*}
where $\lesssim$ only hides a constant depending on $c_0$ in Assumption \ref{ass:DCSBM}.

\paragraph{Proof of (b)}

Based on the fact $\mtx{U}^*_1 = \mtx{M}\mtx{U}^*_1(\mtx{\Lambda}^*_1)^{-1}$, we have
\[
\left\|\mtx{U}^*_1 - \mtx{A}\mtx{U}^*_1 \left(\mtx{\Lambda}^*_1 \right)^{-1}\right\|_{2 \rightarrow \infty} = \left\|(\mtx{M}- \mtx{A})\mtx{U}^*_1 \left(\mtx{\Lambda}^*_1 \right)^{-1}\right\|_{2 \rightarrow \infty}.
\]
By Lemma \ref{lem:population_eigen_vec} and assumption (\ref{eq:sparseness}), applying (\ref{eq:A4_original}) with $\mtx{W} = \mtx{U}^*_1$ yields that with probability $1-O\left(n^{-3}\right)$, 
\begin{align*}
    \left\|(\mtx{A}- \mtx{M})\mtx{U}^*_1(\mtx{\Lambda}^*_1)^{-1}\right\|_{2 \rightarrow \infty} 
    &\leq \left( \underset{1\leq i\leq n}{\max} \left\|(\mtx{A} - \mtx{M})_{i \cdot} \mtx{U}^*_1 \right\|_2  \right) \cdot \left\| \left(\mtx{\Lambda}^*_1 \right)^{-1} \right\| \\
    &\leq C(c_0) \max\left(\theta_{\min}\left\|\mtx{U}^*_1\right\|_F\sqrt{\log n}, \text{~} \left\|\mtx{U}^*_1\right\|_{2 \rightarrow \infty} (\log n)  \right) \cdot \left(\theta_{\min}^2  n \right)^{-1}
    \\
    & \leq C(c_0) \max\left( \left(\theta_{\min} \sqrt{n}\right)^{-1}\left\|\mtx{U}^*_1\right\|_F \sqrt{\frac{\log n}{n}}, \text{~} \left(\theta_{\min} \sqrt{n}\right)^{-2} \left\|\mtx{U}^*_1\right\|_{2 \rightarrow \infty} (\log n) \right)
    \\
    &\lesssim \frac{1}{\sqrt{n \log^5 n}},
\end{align*}
where $\lesssim$ only hides a constant depending on $c_0$ in Assumption \ref{ass:DCSBM}.

With (a) and (b), we have 
\[
\left\|\mtx{U}_1\mtx{Q}^\top - \mtx{U}_1^* \right\|_{2 \rightarrow \infty} \lesssim \frac{1}{\sqrt{ n\log n}} \quad \text{for some~} \mtx{Q} \in \mathcal{O}_{K_1}.
\]
For eigenvalues in group (iii), we similarly have
\[
\left\|\mtx{U}_2\mtx{Q}^\top - \mtx{U}_2^* \right\|_{2 \rightarrow \infty} \lesssim \frac{1}{\sqrt{ n\log n}} \quad \text{for some~} \mtx{Q} \in \mathcal{O}_{K_2}. 
\]
Combining these results yields the second statement in this lemma: With probability $1-O(n^{-3})$, we have
\begin{equation*}
    \left\|\mtx{U}\mtx{Q}^\top - \mtx{U}^* \right\|_{2 \rightarrow \infty} \lesssim \frac{1}{\sqrt{ n\log n}} \quad \text{for some } \mtx{Q} \in \mathcal{O}_{K-1}.
\end{equation*}

To show the first statement, we apply Lemma \ref{thm:row-wise-Abbe} to group (i) with $s=0$ and $r=1$. Note that by Lemma \ref{lem:population_eigen}, we have
\[
\Delta^* = \min\{\lambda_1^* , \lambda_1^* - \lambda_2^* \} \geq C(c_0) \theta_{\min}^2  n
\]
and $\kappa \leq C(c_0)$. Following similar procedures, we can select $\mtx{u}_1$ such that
\begin{enumerate}[label=(\alph*)]
    \item $\left\|\mtx{u}_1 - \left. \mtx{A}\mtx{u}_1^* \middle/\lambda_1^* \right. \right\|_{ \infty} \leq \frac{C(c_0)}{\sqrt{n \log n}}$.
    \item $\left\|\mtx{u}_1^* - \left. \mtx{A}\mtx{u}_1^* \middle/\lambda_1^* \right. \right\|_{ \infty} \leq \frac{C(c_0)}{\sqrt{n \log n}}$.
\end{enumerate}
Therefore, with probability $1-O(n^{-3})$, we have
\begin{equation*}
    \left\|\mtx{u}_1 - \mtx{u}_1^* \right\|_{ \infty}  \leq \frac{C(c_0)}{\sqrt{n \log n}},
\end{equation*}
which proves the first statement.
\end{proof}

\newpage

\section{ADDITIONAL EXPERIMENTS}
\label{sec:additional_exp}

\setcounter{figure}{0}
\renewcommand\thefigure{\thesection.\arabic{figure}}   

This section continues from Section \ref{sec:experiment_synthetic}, presenting additional results of synthetic simulations with larger values of $K$. The generating mechanism of the weighted DCSBM and experiment setup follow from Section \ref{sec:experiment_synthetic}.
We consider the following two settings under the Poisson DCSBM:
\begin{itemize}
    \item \textbf{Setting 1: } $\vct{n}_{all} = (20, 60, 20, 60, 20, 60, 20, 60, 20, 60)$; $K=7,8,9,10$
    \item \textbf{Setting 2: } $\vct{n}_{all} = (30, 60, 90, 30, 60, 90, 30, 60, 90)$; $K=7,8,9$
\end{itemize}
The threshold used in SVPS is $2.02$.
The plots of the accuracy rates of the compared methods are shown in Figure \ref{fig:additional_exp_1}, with figure captions indicating the choice of $(\rho,r)$. As we can see, SVPS still achieves higher accuracy rate, especially when using SCORE for spectral clustering, the performance gap is very significant.

\begin{figure*}[t!]
\begin{subfigure}{.33\textwidth}
  \centering
  \includegraphics[width = 1.0\linewidth, height=0.7 \linewidth]{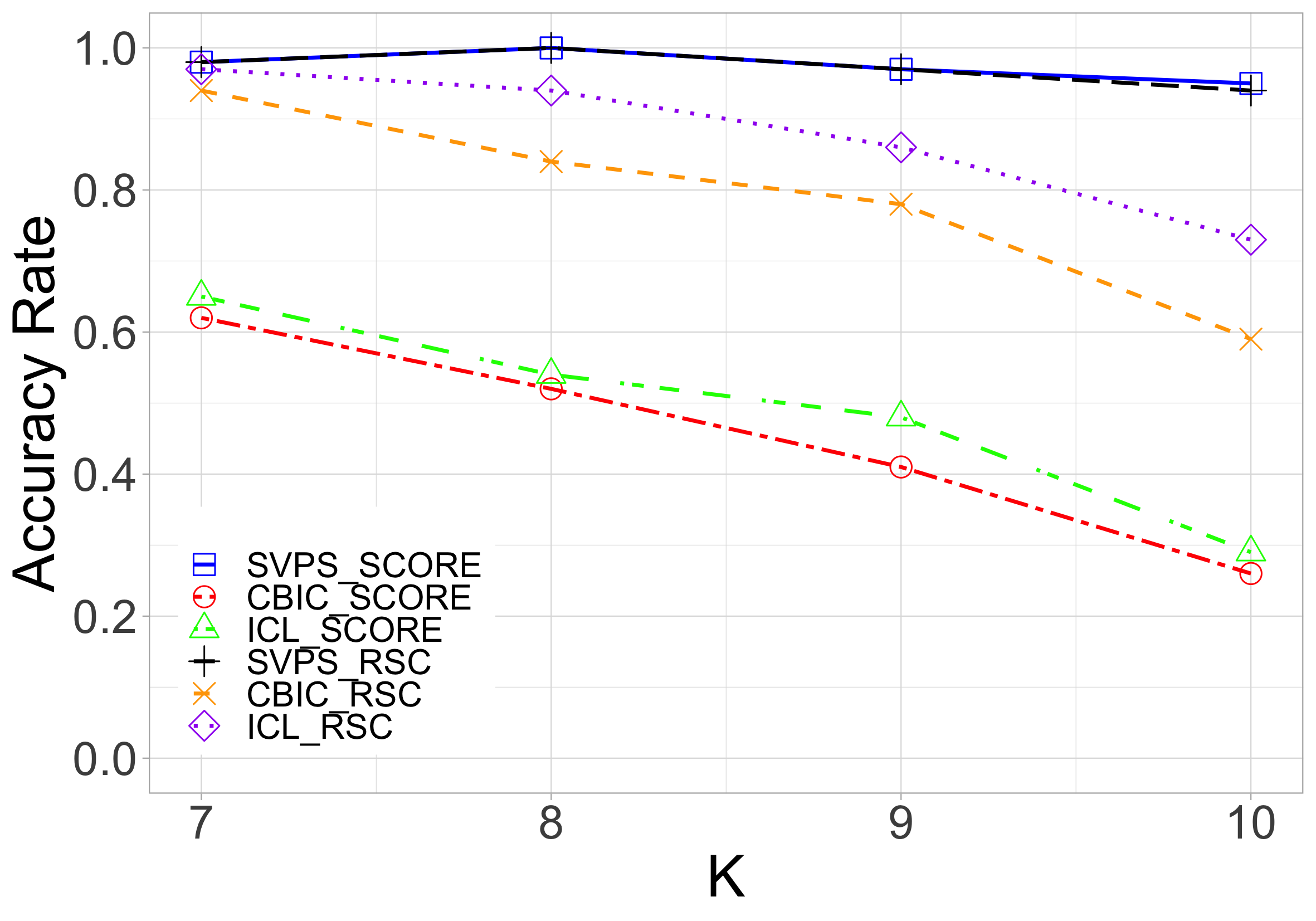}
  \caption{$\rho = 0.2$, $r=4$}
\end{subfigure}%
\begin{subfigure}{.33\textwidth}
  \centering
  \includegraphics[width = 1.0\linewidth, height=0.7 \linewidth]{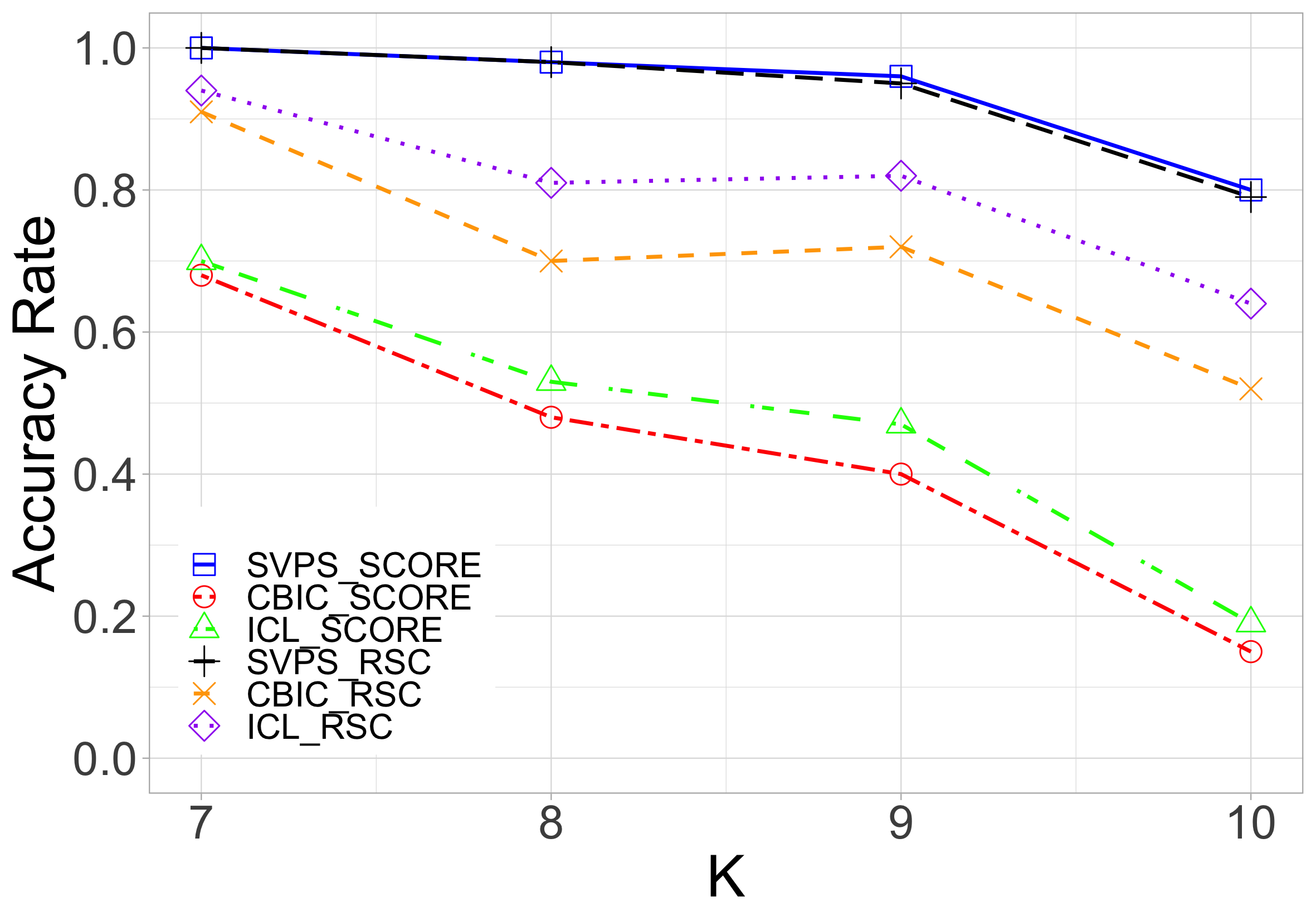}
  \caption{$\rho = 0.3$, $r=3$}
\end{subfigure}
\begin{subfigure}{.33\textwidth}
  \centering
  \includegraphics[width = 1.0\linewidth, height=0.7\linewidth]{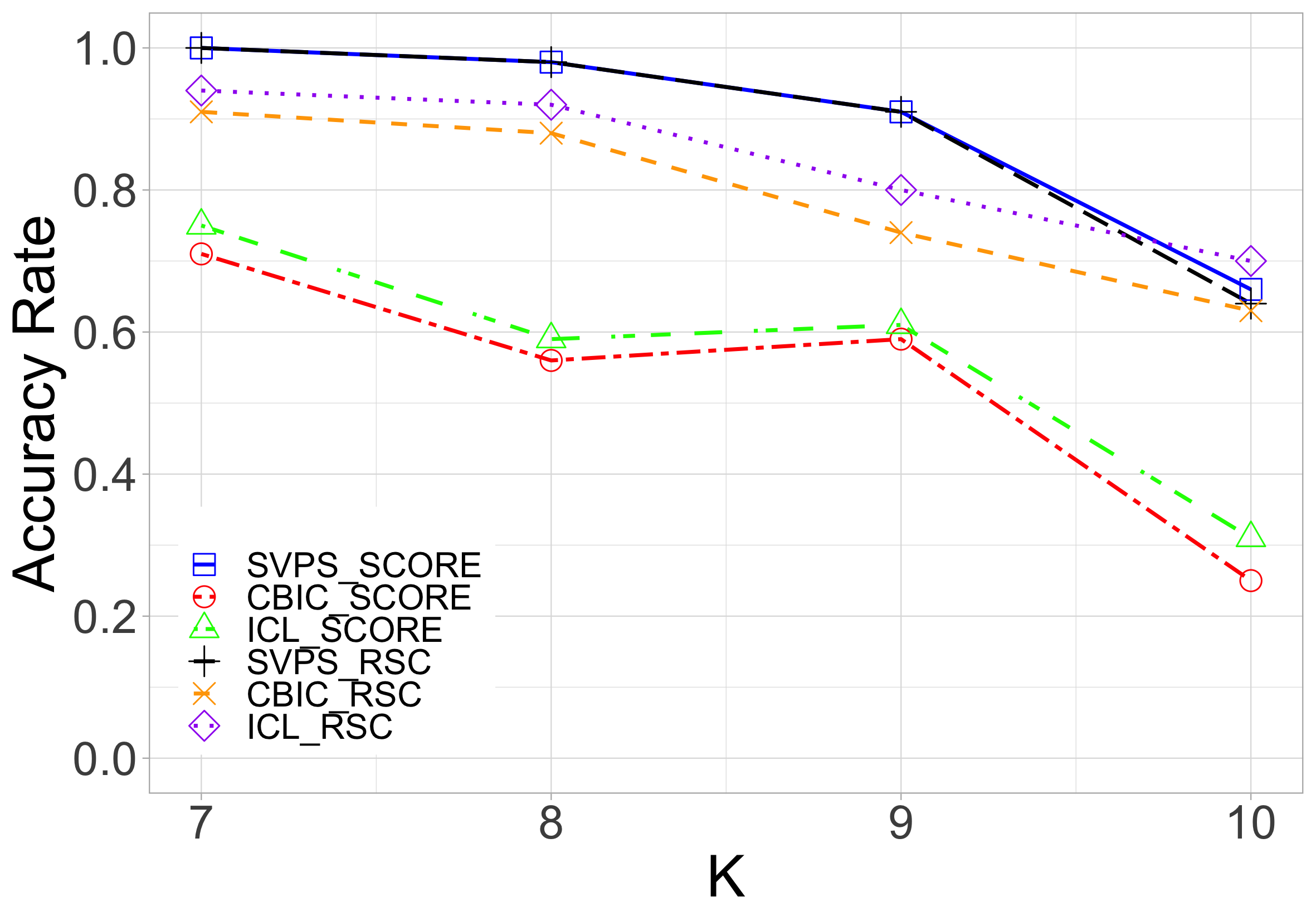}
  \caption{$\rho = 0.6$, $r=2$}
\end{subfigure}

\begin{subfigure}{.33\textwidth}
  \centering
  \includegraphics[width = 1.0\linewidth, height=0.7 \linewidth]{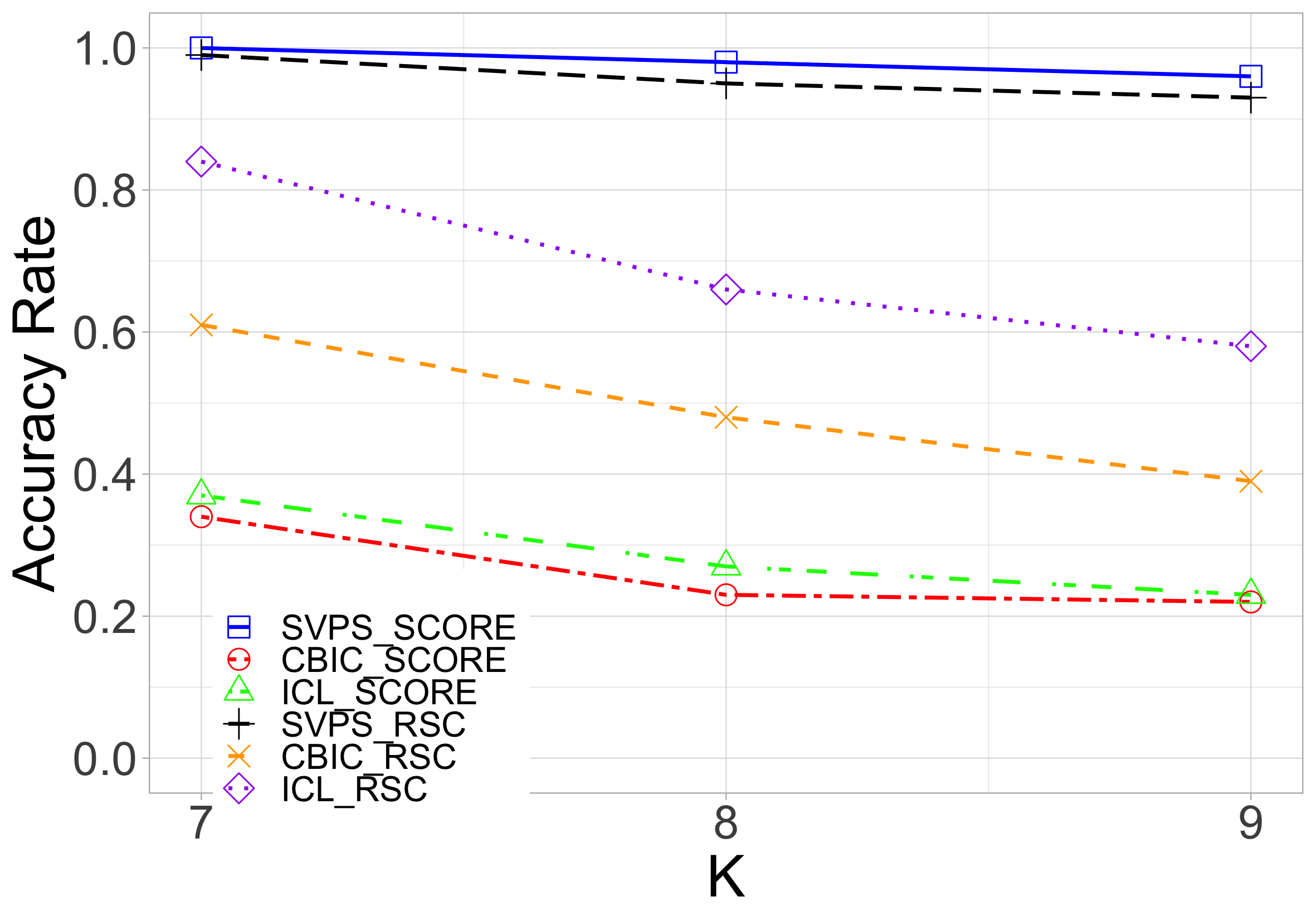}
  \caption{$\rho = 0.09$, $r=4$}
\end{subfigure}%
\begin{subfigure}{.33\textwidth}
  \centering
  \includegraphics[width = 1.0\linewidth, height=0.7 \linewidth]{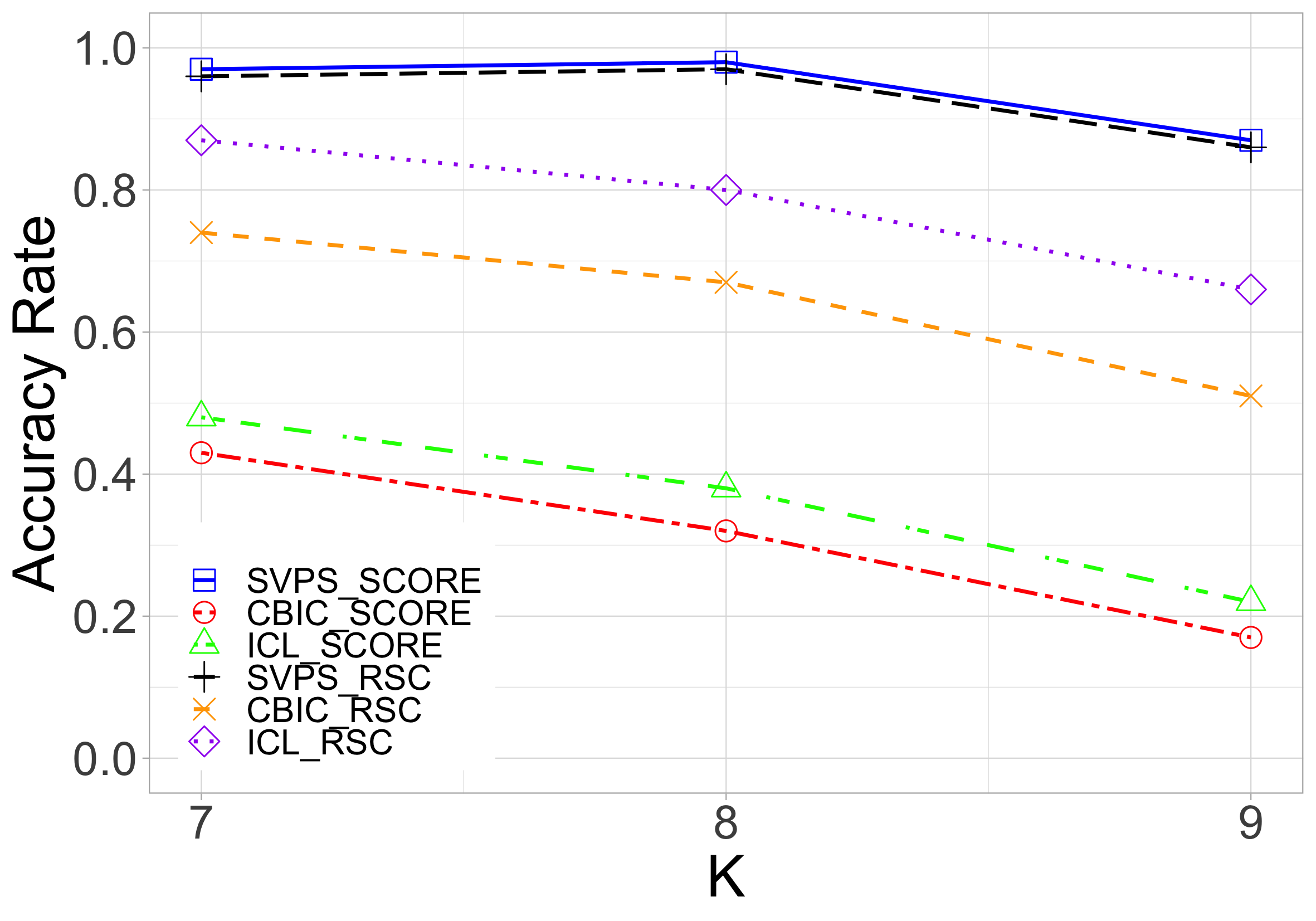}
  \caption{$\rho = 0.15$, $r=3$}
\end{subfigure}
\begin{subfigure}{.33\textwidth}
  \centering
  \includegraphics[width = 1.0\linewidth, height=0.7\linewidth]{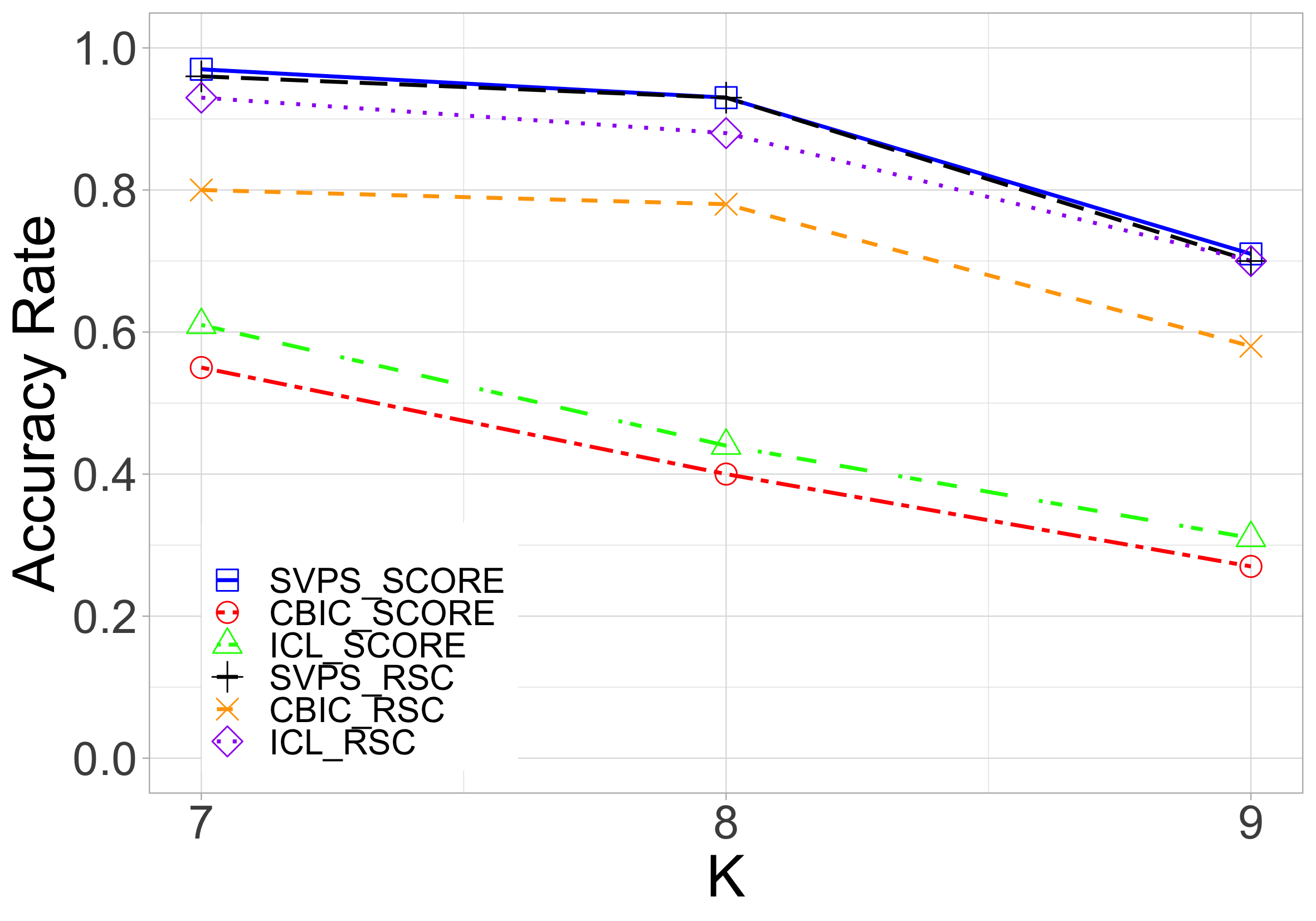}
  \caption{$\rho = 0.3$, $r=2$}
\end{subfigure}

\caption{Accuracy rate plots of experiments in Appendix \ref{sec:additional_exp}. The top and bottom rows correspond to Setting 1 and 2, respectively.}
\label{fig:additional_exp_1}
\end{figure*}

\vfill

\end{document}